\documentclass[11pt]{article}
\usepackage[T1]{fontenc}
\usepackage[utf8]{inputenc}
\usepackage[
letterpaper,
margin=1in
]{geometry}
\usepackage{faktor}
\usepackage{mathtools}
\usepackage{graphicx} 
\usepackage{orcidlink}
\usepackage{amsfonts}
\usepackage{amsmath}
\usepackage{braket}
\usepackage{cleveref}
\usepackage{amsthm}
\usepackage{thmtools}
\usepackage{soul,xcolor}
\usepackage{mathtools}
\usepackage{amssymb}
\usepackage{dsfont}
\usepackage{relsize}

\newcommand{\kg}[1]{}
\newcommand{\knote}[1]{}
\newcommand{\joon}[1]{}
\newcommand{\aer}[1]{}

\usepackage{tikz}
\usetikzlibrary{positioning,fit,decorations.pathreplacing,shapes.geometric, shadows}
\usetikzlibrary{arrows.meta}

\newtheorem{theorem}{Theorem}[section]
\newtheorem*{theorem*}{Theorem}
\newtheorem{definition}[theorem]{Definition}

\newtheorem{proposition}[theorem]{Proposition}
\newtheorem{lemma}[theorem]{Lemma}

\newtheorem{corollary}[theorem]{Corollary}
\newtheorem*{corollary*}{Corollary}

\DeclareMathOperator{\corank}{corank}

\newcommand\PP{\mathbb{P}}
\newcommand\EE{\mathbb{E}}
\newcommand\FF{\mathbb{F}}

\DeclareMathOperator{\id}{id}
\DeclareMathOperator{\GL}{GL}
\newcommand\bc[1]{\left({#1}\right)}
\newcommand\brk[1]{\left[{#1}\right]}

\newcommand\cbc[1]{\left\{#1\right\}}

\newcommand\fs{Q}

\makeatletter
\newcommand*{\centerfloat}{%
  \parindent \z@
  \leftskip \z@ \@plus 1fil \@minus \textwidth
  \rightskip\leftskip
  \parfillskip \z@skip}
\makeatother

\author{
Kenneth Goodenough\thanks{Naturwissenschaftlich-Technische Fakult\"{a}t, Universit\"{a}t Siegen, Walter-Flex-Stra\ss e 3, 57068 Siegen, Germany} \,\orcidlink{0000-0002-1761-0038}
\and
Andrew Landahl\thanks{
             Microsystems Engineering, Science and Applications,
             Sandia National Laboratories,
             Albuquerque, NM, 87185, USA}\, 
             \thanks{QNM-I, Center for Quantum Information and Control, 
             Department of Physics and Astronomy,
             University of New Mexico,
             Albuquerque, NM, 87131, USA}\, \orcidlink{0009-0000-8413-0854}
\and 
Joon Lee\thanks{LIACS, Leiden University} \,\orcidlink{0000-0001-8628-9922}
\and
Antonio Russo\thanks{Center for Computing Research,
             Sandia National Laboratories,
             Albuquerque, NM, 87185, USA}\, \orcidlink{0000-0003-3742-343X}
\and 
Kevin Thompson\footnotemark[5] \,\orcidlink{0000-0001-5669-2200}
}

\title{Optimal Fusion Strategies for Quantum Computation}

\begin{document}
\maketitle
\abstract{Logical fusions are important for a number of tasks in quantum information, such as quantum error correction and quantum repeaters. In the photonic setting one must contend with the fact that physical fusions can fail, i.e. individual qubits are measured in a user-controlled product state basis with some probability, which can lead to a failure on the logical level. The choice of failure basis of each qubit is known as a fusion strategy, and finding good fusion strategies is important for optimizing performance of fusion-based quantum computation. Here we provide a complete characterization when $k=1$ qubits are encoded, and in particular characterize those codes and fusion strategies such that all but one physical fusion can fail, i.e.~\emph{perfect fusion strategies}. In doing so, we recover previously known perfect fusion strategies, and find perfect fusion strategies for quantum parity-check codes, answering an open question. We furthermore show that perfect fusion strategies are generic: random $[[n, 1, d]]$ graph codes admit a perfect fusion strategy with probability exponentially close to $1$.  Additionally we motivate the study of a new graph parameter, namely the maximum degree of a graph at a given vertex taken over all LC-equivalent graphs, by giving a new operationally meaningful interpretation of it.  }
\section{Introduction}
Fusions \cite{browne2005resource, bartolucci2023fusion} form a key component in a variety of contexts: teleportation-based quantum error correction\cite{ewert2016ultrafast}, quantum repeaters~\cite{azuma2015all}, and fusion-based quantum computation~\cite{bartolucci2023fusion} (FBQC).  Fusions can be roughly described as quantum measurements which either ``succeed’’ and implement a Bell state measurement on the underlying qubits, or ``fail’’ and implement a product state measurement on the underlying qubits, each with equal probability.  Using linear optics the “failure basis” can be controlled, i.e.~we can choose for each fusion whether the two photons which are fused are measured in the $X$, $Y$ or $Z$ basis. The failure probability can however be made arbitrarily small~\cite{grice2011arbitrarily, ewert20143, zaidi2013beating} with “boosting”, but the resource requirements for the boosting procedure often lead to poor tradeoffs in overall resource requirements for quantum applications~\cite{li2015resource}.  The failure probability without boosting is generally $1/2$.

Logical fusions have been proposed as a solution to achieve success probability close to $1$ without boosting \cite{bartolucci2023fusion}.  In this context qubits are encoded into some fixed quantum error correcting code (QECC) and physical fusions are implemented between corresponding qubits with the hope of generating an overall logical fusion between the encoded qubits (see \Cref{fig:teleport}).  Each physical fusion either succeeds or fails and the user gets access to measurements of different observables depending on the success/fail outcomes.  A logical fusion is successful if one can infer encoded observables $\bar{X}\bar{X}$ and $\bar{Z}\bar{Z}$ from measurement data.  There is an important design choice here in that one has to choose both the QECC and the set of failure bases $\fs_i$ for each qubit $i$.  For a fixed code we will refer to the choice of failure observables, $\otimes_i \fs_i$ as the fusion strategy. A judicious choice of code and fusion strategy can simplify implementation while improving resilience to errors and fusion failures.  In this paper we will focus on the noiseless setting where fusion failure is the only source of error. 

There is a great deal of freedom in designing how physical fusions fail and what we do with the failure information.  In addition to being able to specify the failure basis, one could also implement adaptive fusion strategies where physical fusions are implemented one at a time and the fusion strategy is altered as a result of the success/failure outcomes \cite{reiss2026optimal} (so-called "adaptive" fusion strategies).  We will consider only fusion strategies where both qubits in a physical fusion are measured in the $X$, $Y$ or $Z$ basis and where the fusion strategy is set before any fusions are implemented, i.e. ``non-adaptive'' strategies.  Adaptive strategies can be strictly better than their adaptive counterparts \cite{reiss2026optimal}, but non-adaptive strategies offer numerous practical advantages over adaptive counterparts including reduced implementation complexity, reduced latency requirements for control/switching, and reduced decoherence/loss.  Indeed, photons must idle in adaptive strategies while classical control adjusts the failure bases, which can be a significant source of error.

There is a very natural question in this setting: what is the minimum number of physical failures that can be tolerated before logical fusion failure? Interestingly, it is well known \cite{reiss2026optimal} that for specific choices of codes one requires only a single successful fusion to guarantee a successful logical fusion.  This seems somewhat surprising when we consider an exemplar problem (\Cref{fig:teleport}). Imagine that we have a single encoded qubit $\ket{\psi}_L =\alpha \ket{\bar{0}} +\beta \ket{\bar{1}}$.  Our goal is to use an encoded Bell state $\ket{\Phi}_{L_A, L_B}\propto  \ket{\bar{0}}_{L_A} \ket{\bar{0}}_{L_B} +\ket{\bar{1}}_{L_A} \ket{\bar{1}}_{L_B}$ to teleport the encoded qubit from the subsystem $L$ to the subsystem $L_B$.  As in standard quantum teleportation, we can accomplish this by doing a logical fusion on the subsystems $L$ and $L_A$, which is implemented through fusions between corresponding qubits.  It is clear that if all fusions fail  then we have destroyed the logical information: in the event that all fusions failed, we have measured single qubit observables on all qubits in $L$. However for most codes there is a choice of fusion strategy where the logical qubit is teleported to $L_B$ if {\it any} fusion succeeds while the remaining qubits failed.  The logical qubit is ``routed'' through any successful fusion and it does not matter which one succeeds.  A fusion strategy in which only a single fusion needs to succeed will be referred to as a {\it perfect strategy} (with respect to $\mathcal{C}$), since it has optimal performance in the noiseless setting. 

\subsection{Prior Work}
Perfect strategies in the non-adaptive setting for specific QECs include the standard $5$-qubit code \cite{reiss2026optimal}, quantum repetition codes \cite{lee2015nearly}, and a small $5$-qubit surface code \cite{schmidt2019efficiencies}. In the adaptive setting there are many strategies known which achieve the optimal performance with respect to fusion failures \cite{reiss2026optimal, pettersson2025long, lee2019fundamental}, as well as bounds on the performance of codes in this setting \cite{lee2019fundamental, hilaire2023linear}.  Additionally, there are recently discovered sufficient conditions for the existence of perfect \emph{adaptive} strategies \cite{reiss2026optimal}, but to our knowledge there are no known simple conditions describing perfect non-adaptive strategies. A criterion which describes necessary and sufficient conditions that an optimal fusion strategy must satisfy would be useful for searching for fusion strategies and finding codes with good performance overall.

\paragraph{Our Results}

In this work we provide an exact characterization of the optimal fusion strategy (given the code) in the non-adaptive and noiseless (lossless) setting for stabilizer codes encoding a single qubit: the optimal $Q$ can always be expressed in terms of the largest non-trivial logical operator of $\mathcal{C}$ that does not contain a non-identity stabilizer as substring.  This yields a complete characterization of perfect strategies for stabilizer codes $\mathcal{C}$ encoding a single qubit: any perfect fusion strategy is a full-weight non-trivial logical operator of $\mathcal{C}$, not containing any stabilizer as a substring.  We confirm that known examples of perfect fusion strategies satisfy our criteria, and use the characterization to find many more perfect strategies. In fact, we show that quantum parity-check codes \cite{shor1995scheme} admit perfect fusion strategies, solving an open problem stated in \cite{schmidt2019efficiencies}.

We phrase the above characterization of perfect fusion strategies for $[[n, 1, d]]$ graph codes. These are codes described by a graph $G$ and a classical error-correcting code $C$ on $\mathbb{F}_2^n$ \cite{cross2008codeword}.  Given $(G, C)$ defining such a code, one can define the progenitor graph \cite{bell2023optimizing, schlingemann2001quantum} $G'$ on $n+1$ vertices: $n$ vertices correspond to the physical qubits, and one vertex of $G'$, the `encoding vertex', corresponds to the single encoded qubit. We show that the maximum degree of the encoding vertex in this graph taken over all graphs which are local-complementation (LC) equivalent, $\Delta_{LC}^{(enc)}(G')$, determines the smallest number of physical fusions that can fail, optimized over all fusion strategies.

We further demonstrate that in \Cref{thm:random_graph_body} that, for a uniformly chosen $G'$ (with $n+1$ vertices) $\Delta_{LC}^{(v)}(G')$ is equal to $n$ with high probability, independent of the vertex $v$. This then implies that random graph codes have perfect fusion strategies with high probability. Our results are congruous with our qualitative observations that many codes have perfect fusion strategies.  Indeed, we have observed that ``good'' (high distance) and ``bad'' (low distance) quantum codes seem to support perfect fusion strategies\footnote{The uniform random graph-code and stabilizer-code distributions induce different distributions on LC-classes, as is easily seen for \(n=2\). Thus, although \(k=1\) random stabilizer codes likely have perfect strategies, our results do not directly imply this.}.

The minimum degree of a graph taken over all LC equivalent graphs is a well-studied graph parameter \cite{javelle2012minimum, cattaneo2015minimum, claudet2024covering}, which has important connections to the complexity of producing graph states \cite{hoyer2006resources}.  In contrast, the \emph{maximum} LC degree is understudied \cite{oumazouz2025classification} and no operationally meaningful correspondence is known.  In this work, we provide the first such correspondence and characterize the maximum LC degree for uniform independent graphs.  These results provide new motivation for the study of this graph parameter over other families and distributions of graphs.  It seems likely that most $n+1$ vertex graphs have maximum  LC degree equal to $n$, so the important question here seems to be understanding when graphs or families of graphs fail to have this property.

While we focus on the exemplar problem of qubit teleportation, it is important to note that the results here also apply more generally to encoded Bell state measurements (BSMs), which are useful for encoded fusions inside a FBQC scheme and other applications.

\section{Notation}

\begin{itemize}
    \item  For a binary vector $v\in \mathbb{F}_2^n$ we define the Hamming weight $|v|=|\{i\in [n]: v_i=1\}|$.
    \item Bold face is used for random variables.  $\mathbf{1}\{A\}$ is the indicator random variable for the event $A$.  All the outcome spaces we are considering are finite.
    \item  As is standard for two events $A$ and $B$ we will use $\mathbb{P}[A, B]$ to denote $\mathbb{P}[A \cap B]$.
    \item  $\mathds{1}$ is the ``all-ones'' vector.  For a set $V$ and a subset $S$ of $V$, $\mathds{1}_S\in \mathbb{F}_2^{|V|}$ is the all ones vector supported on coordinates associated with elements of $S$.
    \item  Given a graph $G=(V, E)$ and a subset $S\subseteq V$ then $G[S]$ is the induced subgraph on $S$.  For a vertex $v\in V$ $N_G(v)$ is the set of neighbors of $v$ in the graph $G$.  If the graph is obvious from context we will omit the subscript.  
    \item  For a matrix $A\in \mathbb{F}_2^{n\times n}$ and a subset $I\subseteq [n]$ we will use $A[I]$ to denote the principal submatrix of $A$ corresponding to the indices $I$.  If the coordinates are associated to elements of a set $U$ then $A[S]$ for $S\subset U$ will denote the principal submatrix of $A$ defined by the coordinates corresponding to $S$.  For two subsets $S$ and $T$, $A[S, T]$ will be the submatrix of $A$ corresponding to rows from $S$ and columns from $T$.
    \item $n$ will generally be the number of physical qubits being considered in an error correcting code.  
    \item $\mathcal{P}_n$ is the group of Pauli operators on $n$ qubits.  $\overline{\mathcal{P}_n}:=\mathcal{P}_n/\langle i \mathbb{I} \rangle$
    \item  We will say that an operator $A\in \mathcal{P}_n$ is ``full weight'' if it is of the form $A=\sigma_1 \otimes \cdots  \otimes \sigma_n$ where none of the $\sigma_i$'s are the identity.  We define $A(i):= \sigma_i$ (for arbitrary $A\in \mathcal{P}_n$).
    \item Let $P$ be a Pauli string, and let $Q$ be any Pauli string that can be obtained from $P$ by setting any entries to $I$. We then say that $Q$ is a substring of $P$, and that $P$ is a superstring of $Q$.  Given a code $\mathcal{C}$ we say that $P$ has a stabilizer substring if it has a substring that is a stabilizer in $\mathcal{C}$.
    \item  For $P=\sigma_1 \otimes ... \otimes \sigma_n$ a Pauli string, $|P|$ is defined as the Pauli Hamming distance, or the number of $\sigma_i$ that are not equal to $\mathbb{I}$.
    \item For any subgroup $B$ of $\mathcal{P}_n$, the dual $B^\perp$ of $B$ is the group of all elements of $\mathcal{P}_n$ commuting with every element of $B$.
    \item  For an arbitrary $A=\sigma_1 \otimes ... \otimes \sigma_n \in \overline{\mathcal{P}_n}$, we define $\hat{A}\equiv \{\gamma_1 \otimes \gamma_2 \otimes ... \otimes \gamma_n\in \overline{\mathcal{P}_n}: \,\, \gamma_i \in \{\sigma_i, \mathbb{I}\}\}$, i.e.~the set of all Pauli strings that are substrings of $A$ up to a phase.  Note that, if $A$ is full weight, $\hat{A}$ contains a full rank stabilizer group corresponding to a product state where qubit $i$ is stabilized by $\sigma_i$. This notation is similar to \cite{bouchet1988graphic}. 
    \item Given a graph $G$, a graph state $\ket{G}$ is a stabilizer state with stabilizers generated by Pauli operators of the form $\Gamma_i:= X_i \prod_{j\in N_G(i)} Z_j$.  We will refer to these generators as the \emph{canonical generators}.  If $v\in \mathbb{F}_2^n$ then $\Gamma_v:=\prod_{j:v_j=1}\Gamma_j$.
    \item Given a graph $G$ (or the associated graph state $\ket{G}$), a local complementation at a vertex $v$ of $G$ replaces the induced subgraph on the neighborhood of $v$ by its complement. In other words, for every unordered pair of neighbors $a, b$ of $v$, remove the edge $ab$ if it exists, and otherwise add it~\cite{hein2006entanglement}.  An edge pivot along an edge $(u, v)\in E$ corresponds to a local complementation at $u$, followed by a local complementation at $v$ and then followed by a local complementation at $u$.  It is equivalent when switching $u$ and $v$.
\end{itemize}

\section{Necessary and sufficient condition for logical failures}
In this section we provide a necessary and sufficient condition for when a logical failure occurs. This criterion will depend on the code $\mathcal{C}$, choice of fusion strategy $Q$, and the subset $W\subseteq V$ of qubits whose associated fusions failed. We prove our characterization in terms of code deformations. Code deformations are normally used to transform one stabilizer code into another by performing measurements. 

As a warmup for the setting of logical fusions, we first state some known results on code deformations.

\subsection{Code deformations}
As before, let $\mathcal{C}$ be an $[[n, k, d]]$ stabilizer code. We will now study the effects different types of measurement will have on the information stored by the code $\mathcal{C}$; in the next subsection we use this to understand the effects of fusion failures.

Stabilizer measurements will always yield a $+1$ outcome if no error has occurred, independent of the encoded state. That is, the probability distribution of all measurement outcomes does not depend on the state that was encoded. Even in the case where Pauli errors did occur, the same syndrome would be observed, independent of which state was encoded. 

The probability distribution when measuring a non-trivial logical operator \emph{does} depend on the encoded state, however. Since information is gained about the encoded state when measuring a non-trivial logical operator, the encoded state loses its coherence.

The idea that measurements that do not reveal logical information preserve logical information is fundamental for code deformation, see~\cite{aasen2023measurement} for an excellent introduction in the context of Floquet codes. Measuring Pauli operators that are not non-trivial logical operators thus transforms the logical information encoded in $\mathcal{C}$ into some $[[n, k, d']]$ code $\mathcal{C}'$. Note that the $k$ parameter remains the same, implying that there is an isometry between the two codes; in other words, the logical information is preserved. The case of stabilizer measurements is an extreme example of this, where $\mathcal{C'}=\mathcal{C}$.

Consider now a state $\ket{\psi_L}$ encoded in $\mathcal{C}$, which gets measured according to some (not necessarily full-weight) Pauli string $P=P_1\otimes \cdots \otimes P_n$, i.e.~qubit $i$ gets measured in the basis of eigenvectors of $P_i$ (if $P_i=\mathbb{I}$ then the qubit is not measured). The probability distribution over all measurement outcomes does not depend on the encoded state, if and only if the probability distribution for each measurement in $\hat{P}$ does not depend on the encoded state. In other words, the encoded information is preserved iff $\hat{P}$ does not contain a non-trivial logical operator. In particular, this means that the post-measurement state is necessarily of the form

\begin{align}
\ket{\psi_L'}\otimes \textrm{tensor product of single qubit stabilizer states}\ ,
\end{align}
such that after an isometry the original state $\ket{\psi_L}$ can be recovered.

\subsection{Moving on to fusion failures}
With the concept of code deformations in mind, we move on to encoded fusions. Here, one takes an encoded state $\ket{\psi_L}$ in some code $\mathcal{C}$, and an encoded Bell pair $\sqrt{2}\ket{\Phi_{L_A,L_B}}=\ket{\bar{0}_{L_A}\bar{0}_{L_B}}+\ket{\bar{1}_{L_A}\bar{1}_{L_B}}$, where both sides of the Bell pair are encoded in the same code. Then, one performs a logical teleportation, teleporting the state $\ket{\psi_L}$ using $\ket{\Phi_{L_A,L_B}}$. In the case of no fusion failures, it is clear that the logical teleportation will have succeeded.

\begin{figure}
\centerfloat

\begin{tikzpicture}[
    scale=0.62,
    qubit/.style={
        circle,
        draw=brown!45!black,
        fill=orange!8!white,
        line width=0.85pt,
        minimum size=5mm,
        inner sep=0pt,
        drop shadow={
            shadow xshift=0.9pt,
            shadow yshift=-0.9pt,
            opacity=0.23
        }
    },
    grayqubit/.style={
        circle,
        draw=gray!55,
        fill=gray!30,
        line width=0.85pt,
        minimum size=5mm,
        inner sep=0pt,
        draw opacity=0.45,
        fill opacity=0.28,
        drop shadow={
            shadow xshift=0.7pt,
            shadow yshift=-0.7pt,
            opacity=0.12
        }
    },
    block/.style={
        draw=brown!50!black,
        rounded corners=5pt,
        line width=0.95pt,
        inner xsep=5pt,
        inner ysep=5pt,
        fill=orange!3!gray!4!white,
        drop shadow={
            shadow xshift=1.2pt,
            shadow yshift=-1.2pt,
            opacity=0.23
        }
    },
    bellblock/.style={
        draw=brown!50!black,
        rounded corners=5pt,
        line width=0.95pt,
        inner xsep=5pt,
        inner ysep=5pt,
        fill=orange!3!gray!4!white,
        drop shadow={
            shadow xshift=1.0pt,
            shadow yshift=-1.0pt,
            opacity=0.23
        }
    },
    logical/.style={
        font=\small,
        align=center
    },
    labelstyle/.style={
        font=\small
    },
    successfusion/.style={
        draw=blue!65!black,
        rounded corners=7pt,
        line width=1.0pt,
        inner xsep=3pt,
        inner ysep=0.9pt,
        fill=blue!35,
        fill opacity=0.18,
        draw opacity=0.75
    },
    failring/.style={
        draw=red!70!black,
        circle,
        line width=1.0pt,
        minimum size=6.5mm,
        inner sep=0pt,
        fill=red!45,
        fill opacity=0.16,
        draw opacity=0.75
    },
    paulilabel/.style={
        font=\scriptsize\bfseries,
        inner sep=0pt
    },
    panelarrow/.style={
        ->,
        >=stealth,
        line width=1.0pt,
        draw=black!70
    }
]

\def\dy{1.2}

\begin{scope}[xshift=0cm]

\foreach \i in {1,...,5} {
    \node[qubit] (pOnepsi\i) at (0,{(3-\i)*\dy}) {};
    \node[labelstyle, left=3pt of pOnepsi\i] {};
}

\foreach \i in {1,...,5} {
    \node[qubit] (pOneA\i) at (2.6,{(3-\i)*\dy}) {};
    \node[labelstyle, above left=-1pt and -1pt of pOneA\i] {};
}

\foreach \i in {1,...,5} {
    \node[qubit] (pOneB\i) at (5.2,{(3-\i)*\dy}) {};
}

\node[block, fit=(pOnepsi1)(pOnepsi5)] (pOnepsibox) {};
\node[bellblock, fit=(pOneA1)(pOneA5)(pOneB1)(pOneB5)] (pOnebellbox) {};

\foreach \i in {1,...,5} {
    \node[qubit] at (pOnepsi\i) {};
    \node[qubit] at (pOneA\i) {};
    \node[qubit] at (pOneB\i) {};
}

\node[logical, below=8pt of pOnepsibox] {$\ket{\psi}_{L}$};
\node[logical, below=8pt of pOnebellbox] {$\ket{\Phi}_{L_A L_B}$};

\end{scope}

\begin{scope}[xshift=9.4cm]

\foreach \i in {1,...,5} {
    \node[qubit] (pTwopsi\i) at (0,{(3-\i)*\dy}) {};
    \node[labelstyle, left=3pt of pTwopsi\i] {};
}

\foreach \i in {1,...,5} {
    \node[qubit] (pTwoA\i) at (2.6,{(3-\i)*\dy}) {};
    \node[labelstyle, above left=-1pt and -1pt of pTwoA\i] {};
}

\foreach \i in {1,...,5} {
    \node[qubit] (pTwoB\i) at (5.2,{(3-\i)*\dy}) {};
}

\node[block, fit=(pTwopsi1)(pTwopsi5)] (pTwopsibox) {};
\node[bellblock, fit=(pTwoA1)(pTwoA5)(pTwoB1)(pTwoB5)] (pTwobellbox) {};

\foreach \i in {1,3,5} {
    \node[successfusion, fit=(pTwopsi\i)(pTwoA\i)] {};
}

\node[failring] at (pTwopsi2) {};
\node[failring] at (pTwoA2) {};
\node[failring] at (pTwopsi4) {};
\node[failring] at (pTwoA4) {};

\foreach \i in {1,...,5} {
    \node[qubit] at (pTwopsi\i) {};
    \node[qubit] at (pTwoA\i) {};
    \node[qubit] at (pTwoB\i) {};
}

\node[failring] at (pTwopsi2) {};
\node[failring] at (pTwoA2) {};
\node[paulilabel] at (pTwopsi2) {$Z$};
\node[paulilabel] at (pTwoA2) {$Z$};

\node[failring] at (pTwopsi4) {};
\node[failring] at (pTwoA4) {};
\node[paulilabel] at (pTwopsi4) {$X$};
\node[paulilabel] at (pTwoA4) {$X$};

\end{scope}

\begin{scope}[xshift=18.8cm]

\foreach \i in {1,...,5} {
    \node[qubit] (pThreepsi\i) at (0,{(3-\i)*\dy}) {};
    \node[labelstyle, left=3pt of pThreepsi\i] {};
}

\foreach \i in {1,...,5} {
    \node[qubit] (pThreeA\i) at (2.6,{(3-\i)*\dy}) {};
    \node[labelstyle, above left=-1pt and -1pt of pThreeA\i] {};
}

\foreach \i in {1,...,5} {
    \node[qubit] (pThreeB\i) at (5.2,{(3-\i)*\dy}) {};
}

\node[block, fit=(pThreepsi1)(pThreepsi5)] (pThreepsibox) {};
\node[bellblock, fit=(pThreeA1)(pThreeA5)(pThreeB1)(pThreeB5)] (pThreebellbox) {};

\foreach \i in {1,3,5} {
    \node[successfusion, fit=(pThreepsi\i)(pThreeA\i)] {};
}

\foreach \i in {1,3,5} {
    \node[qubit] at (pThreepsi\i) {};
    \node[qubit] at (pThreeA\i) {};
    \node[qubit] at (pThreeB\i) {};
}

\foreach \i in {2,4} {
    \node[grayqubit] at (pThreepsi\i) {};
    \node[grayqubit] at (pThreeA\i) {};
    \node[qubit] at (pThreeB\i) {};
}

\node[logical, below=8pt of pThreepsibox] {$\ket{\phi}_{L'}$};
\node[logical, below=8pt of pThreebellbox] {$\ket{\Phi'}_{L_{A} L_B}'$};

\end{scope}

\draw[panelarrow]
    ([xshift=6mm]pOnebellbox.east) -- ([xshift=-6mm]pTwopsibox.west);

\draw[panelarrow]
    ([xshift=6mm]pTwobellbox.east) -- ([xshift=-6mm]pThreepsibox.west);

\end{tikzpicture}
\caption{An encoded single qubit quantum state undergoes $\sigma$-fusions for each pair of qubits in $(L, L_A)$.  The type of each fusion is labeled as a subscript of the blue edges on the leftmost figure.  The successful fusions implement a Bell state measurement on the corresponding qubits while a failure implements single qubit measurements in a basis determined by the type of the fusion.  If enough fusions succeed then the logical state on $L$ is teleported to $L_B$ (depicted on the right side of the figure). }\label{fig:teleport}
\end{figure}
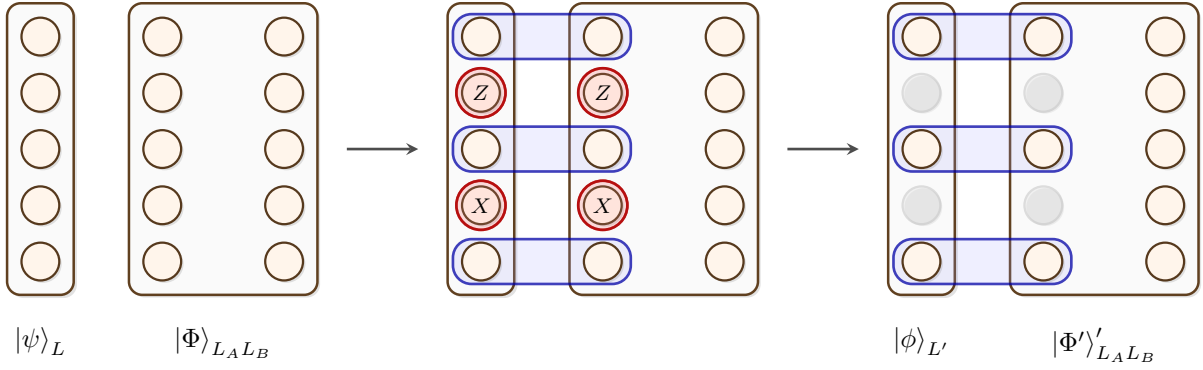

In the general case some subset of the fusions $W$ will fail (see \Cref{fig:teleport}). A failed fusion on $i\in W$ results in the corresponding qubits in $L$ and $L_A$ being measured in the $Q_i$ basis, while the fusion measurements on the remaining qubits succeed. In what follows, it will be convenient to regard the measurements associated with the fusion failures as occurring first, followed by the successful fusion measurements. We are free to change the order of these measurements, since the measurements clearly commute.

Denote by $Q[W]$ the Pauli string that equals $Q$ on the indices $i\in W$, and the identity elsewhere. If a failure now occurs on a subset $W$ of qubits, both the state $\ket{\psi_L}$ and $\ket{\Phi_{L_A,L_B}}$ are measured according to $Q[W]$.
From the preceding section, if $Q[W]$ reveals no information about the encoded state, then after the failures the states are now of the form

\begin{align}
\ket{\psi_{L'}}\otimes \textrm{tensor product of single qubit stabilizer states, }
\end{align}
and
\begin{align}
\propto\left(\ket{\bar{0}_{L_A'}\bar{0}_{
L_B}}+\ket{\bar{1}_{L'_A}\bar{1}_{
L_B}}\right)\otimes \textrm{tensor product of single qubit stabilizer states}\ .
\end{align}
That is, the initial logical information $\ket{\psi_L}$ is now encoded as $\ket{\psi_{L'}}$ in some different code $\mathcal{C}'$, and similarly for the `left' side of $\left(\ket{\bar{0}_{L'_A}\bar{0}_{
L_B}}+\ket{\bar{1}_{L'_A}\bar{1}_{
L_B}}\right)$. Teleporting now effectively re-encodes the state $\ket{\psi_{L'}}$ into $\ket{\psi_{L_B}}$. Note that we might need to perform some corrections, corresponding to the outcomes of our successful fusions. These corrections are always logical Pauli operations, which can always be implemented fault-tolerantly with stabilizer codes.

This leads us to the following proposition.

\begin{proposition}\label{prop:failure_condition}
Let $\mathcal{C}$ be an $[[n, k, d]]$ stabilizer code, and let $Q$ be a fusion strategy. Then fusion failures on a subset $W$ of qubits lead to a logical failure if and only if $\hat{Q}[W]$ contains a non-trivial logical operator of $\mathcal{C}$. 
\end{proposition}

The smallest $|W|$ such that $\hat{Q}[W]$ contains a non-trivial logical operator we call the \emph{failure distance} of $\mathcal{C}$ with respect to $Q$. The \emph{fusion distance} of $\mathcal{C}$ is the \emph{failure distance} maximized over all fusion strategies.

The following two lemmas will be useful.

\begin{lemma}\label{lemma:dimension}
Let $\mathcal{C}$ be an $[[n, k, d]]$ stabilizer code with stabilizer group $\mathcal{S}$, and let $\mathcal{S}^\perp$ be its dual. Let $Q$ be a full-weight Pauli string on $n$ qubits. Then 

\begin{align}
\dim\left(\hat{Q}\cap \mathcal{S}^\perp\right) - \dim\left(\hat{Q}\cap \mathcal{S}\right) = k ,
\end{align}
where here we interpret $\hat{Q}\cap \mathcal{S}^\perp$ and $\hat{Q}\cap \mathcal{S}$ as subspaces of $\hat{Q}\cong\mathbb{F}_2^n$.
\end{lemma}
\begin{proof}
See Appendix \ref{section:proof_lemma}.

\end{proof}

\begin{lemma}\label{lemma:bouchet}
Let $\mathcal{C}$ be an $[[n, k, d]]$ stabilizer code with stabilizer group $\mathcal{S}$, and let $P$ be any (not necessarily full-weight) Pauli string on $n$ qubits. If $\dim(\hat{P}\cap \mathcal{S})=m$, then there exists a full-weight superstring $Q$ of $P$ such that $\dim(\hat{Q}\cap \mathcal{S})=m$.
\end{lemma}
\begin{proof}
See Appendix \ref{section:proof_lemma}.
\end{proof}

A non-trivial logical operator $P$ such that no proper substring is a non-trivial logical operator we call a \emph{minimal} logical operator. From proposition \ref{prop:failure_condition} and the above lemma, we can characterize the fusion distance of all $[[n, 1, d]]$ codes in terms of minimal logical operators.

\begin{theorem}\label{cor:opt_strategy}
    Let $\mathcal{C}$ be an $[[n, 1, d]]$ stabilizer code $\mathcal{C}$. The fusion distance of $\mathcal{C}$ equals the size of the largest minimal logical operator of $\mathcal{C}$.
\end{theorem}
\begin{proof}
For a fusion strategy $Q$, let $N(Q):=\left(\hat{Q}\cap \mathcal{S}^\perp\right)\setminus \left(\hat{Q}\cap \mathcal{S}\right)$ be the set of all non-trivial logical operators in $\hat{Q}$. From the characterization in Proposition \ref{prop:failure_condition}, the failure distance of $\mathcal{C}$ with respect to $Q$ is the smallest weight element in $N(Q)$. $N(Q)$ is always nonempty by \Cref{lemma:dimension}.  A minimum weight element of $N(Q)$ is a minimal logical operator so the fusion distance is upper bounded by the maximum possible  weight of a minimal logical operator.  

The goal is to find a $Q$ such that $N(Q)=\lbrace{P\rbrace}$, where $P$ is a maximum weight minimal logical operator. The failure distance with respect to such a $Q$ (if it exists) would equal the fusion distance of the code.

Note that $\dim(\hat{Q}\cap \mathcal{S})=0$, together with \Cref{lemma:dimension}, implies that $\dim(\hat{Q}\cap \mathcal{S}^\perp)=1$ (since $k=1$).  This means that any full-weight Pauli string $Q$ such that $\hat{Q}$ does not contain any non-trivial stabilizer satisfies $|N(Q)|=1$.

In other words, if $Q$ is a superstring of a maximum weight minimal logical operator $P$, and $\dim (\hat{Q}\cap \mathcal{S})=0$ (i.e.~$\hat{Q}$ does not contain any non-trivial stabilizers), then $N(Q)=\lbrace P\rbrace$, as desired. But now note that $\dim (\hat{P}\cap \mathcal{S})=0$, so that by Lemma \ref{lemma:bouchet} there exists such a full-weight $Q$.

\end{proof}

 Theorem \ref{cor:opt_strategy} gives an immediate characterization of the $k=1$ codes that admit perfect strategies.
\begin{corollary}\label{cor:perfect_strategy}
    Let $\mathcal{C}$ be an $[[n, 1, d]]$ stabilizer code $\mathcal{C}$.  Then $\mathcal{C}$ admits a perfect strategy iff $\mathcal{C}$ has a full-weight non-trivial logical operator not containing any non-trivial stabilizer as a substring.
\end{corollary}

\section{Fusion strategies for graph codes}\label{sec:graph_codes}

An $[[n, 1, d]]$ stabilizer code is defined by a stabilizer state $\ket{\phi}$ on $n+1$ qubits (the so-called progenitor graph \cite{bell2023optimizing}) which ``encodes'' both the stabilizers of the code as well as the logical operators~\cite{khesin2025universal}. The $n+1$ qubits are partitioned into $n$ qubits (corresponding to the physical qubits and labeled by $[n]$), and an \emph{encoding qubit} $e$\footnote{For technical reasons, we require furthermore that $e$ is incident to at least one edge~\cite{goodenough2024near}.}
The logical operators of the corresponding stabilizer code are described by the stabilizers $s$ of $\ket{\phi}$. In particular, every stabilizer of $\ket{\phi}$ can be written as a tensor product of two Pauli strings acting on qubit $e$ and $[n]$, 
i.e.~$s=s_{\textrm{log}}\otimes s_{\textrm{phys}}$ where $s_\textrm{log}$ acts on qubit $e$ and $s_{\textrm{phys}}$ acts on $[n]$. Under this identification, the logical operators are given by all operators of the form $s_{\textrm{phys}}$, which act as $s_{\textrm{log}}$ on the encoded qubits~\cite{khesin2025universal, goodenough2024bipartite}.

It is known that every stabilizer code is equivalent to a \emph{graph code}. That is, up to a relabeling of logical operators and local Clifford (LC) rotations, $\ket{\phi}$ can always be chosen to be a graph state $\ket{G'}$. In this case, the operator $\Gamma_{e}$ is of the form $X_{e} \otimes \prod_{j \in N(e)} Z_j=X_{e} \otimes \bar{X}$, i.e.~\emph{graph codes} have logical $X$ operators implemented by $Z$ operators, where the $Z$ operators correspond to the neighborhood of $e$.  It is further known that graph states are equivalent under local Clifford operations iff they are equivalent under a sequence of local complements \cite{van2004graphical}.  A local complement of a graph at a vertex $v$ is obtained by toggling all the edges in the neighborhood of $v$.  We will use LC to indicate both local Clifford and local complement.  This is well motivated since the set of graph states that are local Clifford (LC) equivalent to a given graph state is the same as the set of graphs are local-complementation (LC) equivalent to the corresponding graph.   

With this in mind, we define the following quantity.

\begin{definition}
    For a vertex $v\in V$, $\Delta_{LC}^{(v)}(G)$ is the maximum degree of vertex $v$ taken over all graphs that are $LC$-equivalent to $G$.
\end{definition}

Note that $\prod_{j\in N(e)}Z_j$ is a minimal logical operator for the code. Since local complementations correspond to single-qubit Cliffords~\cite{van2004graphical}, the possible neighborhoods of $e$ after any sequence of local complementations always form the supports of minimal logical operators. As such $\Delta_{LC}^{(e)}(G')$ is a lower bound on the fusion distance. The following lemma allows us to show that any minimal logical operator arises as the neighborhood of $e$, demonstrating that this lower bound is, in fact, tight.

\begin{lemma}\label{lemma:circuit_to_nhood_stabilizer}
A stabilizer $S$ of a stabilizer state is called a \emph{minimal stabilizer} if $S$ is not the identity, and no proper substring of $S$ is a non-identity stabilizer.
Let $\ket{\phi}$ be any stabilizer state, let $S$ be any of its minimal stabilizers, and fix an arbitrary $v\in \textrm{supp}(S)$. Then $\ket{\phi}$ is locally equivalent to a graph state $\ket{G}$ such that vertex $v$ has closed neighborhood equal to $\textrm{supp}(S)$.
\end{lemma}
\begin{proof}
See Appendix \ref{section:proof_lemma}. A similar statement using different terminology can be found in~\cite{brijder2015isotropic}, see Theorem 5.
\end{proof}

\begin{theorem}\label{thm:maximumdegree}
    Let $\mathcal{C}$ be a graph code defined with progenitor graph $G'$ and encoding vertex $e$.  The fusion distance of $\mathcal{C}$ equals $\Delta_{LC}^{(e)}(G')$.
\end{theorem}
\begin{proof}
From the discussion at the beginning of the section, the (not necessarily non-trivial/minimal) logical operators of $\mathcal{C}$ are given by the stabilizers of $\ket{G'}$.
The minimal logical operators $P$ of $\mathcal{C}$ correspond then precisely to the minimal stabilizers of $\ket{G'}$ that have support on $e$. Since by \Cref{cor:opt_strategy} the fusion distance is given by a maximum over the weights of minimal logical operators, it suffices to find the maximum weight of a minimal stabilizer that has support on $e$. From Lemma \ref{lemma:circuit_to_nhood_stabilizer} every such minimal stabilizer is, up to single-qubit Cliffords, of the form $\Gamma_e$ of some graph state $\ket{G''}$ LC-equivalent to $\ket{G'}$.
\end{proof}

\subsection{Almost all graph codes have perfect fusion strategies}

In the previous sections we characterized (graph) codes admitting perfect fusion strategies, and showed how previous examples in the literature indeed satisfy this criterion. Since perfect fusion strategies are highly desirable, it is natural to wonder how generic perfect fusion strategies are. In this work we demonstrate that a uniformly random graph code has a perfect strategy with probability exponentially close to $1$. 

Uniform random graph codes correspond to uniform random progenitor graphs.  Recall that the graph code has a perfect strategy iff there is a full weight stabilizer of the progenitor graph state with no other stabilizer of the progenitor graph as a substring.  Therefore, to demonstrate our theorem we want to show the existence of such a stabilizer.  First, recall that a cannonical generator of a graph state is a stabilizer of the form $X_i \prod_{j\in N_G(i)} Z_j$.  The product of all the canonical generators of the progenitor graph, denoted as $\Phi$, is always a full weight stabilizer since it has either an $X$ or a $Y$ on every vertex.  

We can understand when $\Phi$ has a stabilizer substring by considering the binary Laplacian.  The binary Laplacian is simply defined as taking the standard Laplacian$\mod 2$ each entry: $L_{i,i}$ is the number of neighbors of vertex $i$ in the graph$\mod 2$ and $L_{i,j}$ for $i\neq j$ is $1$ exactly when $i$ and $j$ are connected.  Let us fix a set $S\neq \emptyset, [n+1]$ and define the vector $\vec{w}=L \mathds{1}_S$.  For $i$ a vertex, $\vec{w}_i$ gives the number of edges in $G'$ adjacent to $i$ which cross the cut naturally defined by the set $S$.  So, e.g., if $i\in S$ then $\vec{w}_i$ would be the number of neighbors of $i$ in the graph which are not in $S$ $\mod 2$.  It follows that $\mathds{1}_S$ is in the null space of $L$ precisely when every vertex in $S$ has an even number of neighbors not in $S$ and when each vertex not in $S$ has an even number of neighbors in $S$.  Note that $L\mathds{1}_{[n+1]}=0$ always, so $L\mathds{1}_S=0 \Leftrightarrow L\mathds{1}_{[n+1]\setminus S}=0$. 

Now let us once again consider $\Phi$.  This stabilizer contains a stabilizer substring exactly when there exists some stabilizer $\Gamma_S$ which matches $\Phi$ on $S$ but is the identity on $[n+1]\setminus S$.  For this to be the case, every vertex in $[n+1]\setminus S$ must have an even number of neighbors in $S$.  Similarly, since $\Gamma_S$ is a substring then $\Gamma_{[n+1]\setminus S}$ must also be a substring.  It follows that every vertex in $S$ must have an even number of neighbors in $[n+1] \setminus S$.  So, we see that $\Phi$ has a stabilizer substring exactly when $L\mathds{1}_S=0$.  

With this observation in hand, a naive approach to showing that a random code has a full weight substring without stabilizer substring is showing that the only null vector of the Laplacian is $\mathds{1}_{[n+1]}$, or equivalently that the rank is $n$, with high probability.  Unfortunately, the Laplacian has maximal rank $n$ with only probability $\approx 0.42$\footnote{The initial observation relating the cut condition to the binary graph Laplacian was suggested by ChatGPT. The subsequent extension to subgraph Laplacians and its application to quantum graph codes were developed by the authors.\label{ft:AI}}.  The key to the proof is considering Laplacians of subgraphs (see \Cref{sec:part_lap}).  

Let $\Gamma_S$ be a full weight stabilizer of $G'$, let $\ket{G[S]}$ be the graph state of the induced graph $G[S]$ and let $\Phi_S$ be the restriction of $\Gamma_S$ to the qubits in $S$.  If $\Gamma_S$ has a stabilizer substring then it is easy to see that $\Phi_S$ has a stabilizer substring for $\ket{G[S]}$.  Since $\Phi_S$ corresponds exactly to the product of canonical generators described above, this implies that the binary Laplacian of the induced subgraph has a null vector other than $\mathds{1}_S$.  Similarly, if there is a null vector of the induced Laplacian $\mathds{1}_R$ other than $\mathds{1}_S$ then since $\Gamma_S$ is known to be full weight on the overall graph $\Gamma_R$ can be extended to a stabilizer substring on all of $G'$ (see \Cref{lem:part_lap}). 

Subsequently, demonstrating that random progenitor graphs have full weight stabilizer with no stabilizer substring reduces to demonstrating that there is a subset of the vertices $S$ such that the induced Laplacian has rank $|S|-1$ and such that $\Gamma_S$ is full weight.  

\begin{theorem}\label{thm:random_graph_body}
    Let $G'$ be a uniform random graph on $n+1$ vertices.  There exists a subset $S$ such that $\Gamma_S$ is full weight and $L_S$ (the induced binary Laplacian) has rank $|S|-1$ with probability at least $1-O(2^{-\delta n})$ for some constant $\delta > 0$.
\end{theorem}

\begin{proof}
See \Cref{thm:gen_perfect_technical}.
\end{proof}

\begin{corollary}
    A randomly sampled graph code encoding a single qubit has a perfect strategy with probability at least  $1-O(2^{-\delta n})$ for some constant $\delta > 0$
\end{corollary}

\section{Examples}\label{sec:examples}
Armed with the above characterization (\Cref{cor:opt_strategy}), it is easy to explain known examples of perfect strategies while also demonstrating new ones.  Indeed, \cite{lee2015nearly} considers the repetition code where $\ket{0}=\ket{+}^{\otimes n}$ and $\ket{1}=\ket{-}^{\otimes n}$ and notes that $Z^{\otimes n}$ is a perfect strategy.  This is simply a representative of $\bar{X}$ and it is easy to see that it contains no stabilizer substring since the repetition code contains only $X$ operators in the stabilizer.  Another previously known example is the five-qubit code \cite{laflamme1996perfect}. The authors of \cite{reiss2026optimal}  note that one can use either $X^{\otimes 5}$ or $Z^{\otimes 5}$ ($Y^{\otimes 5}$ also works) to obtain perfect strategies.  These are both logical operators for the code, and it can be verified that they have no stabilizer substring (any non-full weight $X$ or $Z$ string will anti-commute with one of the stabilizers) so our characterization immediately implies that these are perfect strategies.  The $surface(2, 2)$ code in \cite{schmidt2019efficiencies} also has a perfect strategy which fits our characterization.  Indeed the stabilizer group is 
$$
\mathcal{S} \;=\; \bigl\langle\,Z_1 Z_3 Z_4,\, Z_2 Z_3 Z_5,\,X_1 X_2 X_3,\, X_3 X_4 X_5 \,\bigr\rangle ,
$$
and the perfect strategy they found is 
$$
X_1 Z_2 Y_3 Z_4 X_5
\;=\; Y_1 Z_2 X_4 \cdot\, Z_1 Z_3 Z_4 \cdot  X_3 X_4 X_5=\bar{Y},
$$
where $Y_1 Z_2 X_4$ is a ``canonical representative'' of $\bar{Y}$.

For the quantum parity check (QPC) code, the authors of \cite{schmidt2019efficiencies} stated the existence of a perfect strategy as an open problem; we are able to provide an example of one for a QPC code of any size.  The simplest example is $QPC(2, 2)$ which is a quantum code on $4$ qubits with the stabilizer group equal to $\langle XXXX, ZZII, IIZZ \rangle $.  Representatives of logical operators for this code are $\bar{X}=XXII$ and $\bar{Z}=IZIZ$.  We can see that $XYXY$ is a full weight logical $\bar{Z}$ representative, not containing any stabilizer as a substring. As such, $QPC(2, 2)$ has a perfect fusion strategy with $Q = XYXY$.

In general, the code $QPC(n, m)$ is defined as follows: There are $nm$ qubits where the qubits are organized into $n$ rows of $m$ qubits each.  We can index the qubits by row and position number so that the operator $\sigma_{(i, j)}$ corresponds to operator $\sigma$ acting on the qubit at row $i$ and column $j$. The stabilizer group is generated by tensor products of $X$ operators acting on adjacent rows, $X_{(i, 1)} X_{(i, 2)}... X_{(i, m) } X_{(i+1, 1)} X_{(i+1, 2)}... X_{(i+1, m)} $, as well as pairs of $Z$ operators acting on adjacent qubits inside the same row, $Z_{(i, j)} Z_{(i, j+1)}$.  The number of algebraically independent generators for the code is $(m-1)n+n-1=nm-1$ so these codes encode exactly one logical qubit.  We can take representatives of $\bar{X}$ and $\bar{Z}$ to be $X_{(1, 1)} X_{(1, 2)}... X_{(1, m)}$ and $Z_{(1,1)} Z_{(2, 1)}... Z_{(n, 1)}$ respectively.  

An optimal fusion strategy here is similar to the $QPC(2, 2)$ case, it is the tensor product of operators where in each row every qubit receives an $X$ operator except the last one which receives a $Y$ operator: $X_{(1, 1)} X_{(1, 2)}... X_{(1, m-1)}Y_{(1, m)} X_{(2, 1)}... X_{(2, m-1)} Y_{(2, m)}... Y_{(n, m)}$.  Depending on whether $n$ is even or odd, the fusion strategy is a representative of $\bar{Z}$ or $\bar{Y}$ respectively.  We can easily see this operator contains no stabilizer substring by the following:  Any pure $X$ type stabilizer would have an entire row populated by $X$ operators and we can see there is no substring of this form.  Note also that any other stabilizer would require an even number of $Z$ operators in a row or an even number of $Y$ operators in a row.  We cannot ``read off'' a substring here with an even number of $Y$ or $Z$ operators in a row.

In our numerics we came across several codes without perfect fusion strategies.  The smallest possible example without a perfect strategy and with a connected progenitor graph (if the graph is disconnected it is easy to see there cannot be any perfect strategy) is a code on $5$ qubits and we can describe an explicit generating set as $\langle Z_1 X_2 Z_4, Z_1 X_3 Z_4, Z_2 Z_3 X_4 Z_5, Z_4 X_5 \rangle $.  While this code has no perfect non-adaptive strategy, we were able to confirm the existence of an adaptive strategy which requires only a single fusion success to obtain logical success.  This fits well with existing literature \cite{reiss2026optimal} which indicates that there are codes with adaptive strategies that are strictly better than all non-adaptive strategies in the lossless setting.

\section{Discussion}

A promising direction for future research is the $k>1$ case. In this case, strategies which succeed with only a single fusion success cannot be possible. Our preliminary numerics suggest that the $k=2$ case is richer, for example, the codes we tested do not have strategies where any two fusion successes are sufficient.  
We leave open the question of which values of $n$, $k$ and $t$ are achievable with random codes/graphs with high probability.  Another open direction is looking for error-resistant fusion strategies \emph{subject to the constraint} that the fusion strategy is perfect.  Our results indicate that most codes for $k=1$ have perfect strategies, and given a code it is easy to identify if a given fusion strategy is perfect using our criteria, so it is natural to use ``perfectness'' of the fusion strategy as a constraint in the optimization of logical fusions with the goal of minimizing the effects of some other common error source, i.e. loss.  

Early on in this project we formulated the conjecture that for $k=1$ a connected progenitor graph $G'$ indicated the existence of a perfect strategy.  This was false and the example of the code given in \Cref{sec:examples} is such an example.  However, this code does have a perfect adaptive strategy.  So, a natural conjecture and direction of future work is that all $[[n, 1, d]]$ codes with connected progenitor graphs have perfect adaptive strategies. 

The existence of a fusion strategy with a given failure distance is not something easy to check or determine {\it a priori}.  Indeed, related problems are generally NP-hard at least~\cite{kapshikar2023hardness}.  We leave open the complexity of finding an optimal fusion strategy given a code (or equivalently finding the maximum degree of a vertex in a graph taken over all LC-equivalent graphs).  It is encouraging to note that our results indicate that the maximum LC degree behaves very differently from the minimum LC degree (the distance of the underlying code) for random graphs, which could indicate a complexity difference between the two.

\section*{Acknowledgments}
KG acknowledges the support from the Alexander von Humboldt Foundation.

JL acknowledges the support from the Dutch Research Council (NWO) under the project Boosting the Search for New Quantum Algorithms with AI (BoostQA), file number NGF.1623.23.033, research programme Quantum Technologie 2023.

AL, AR and KT are supported by the Error-corrected Photonic Integrated Qubits (EPIQ) project.  Sandia National Laboratories is a multimission laboratory managed and operated by National Technology and Engineering Solutions of Sandia, LLC, a wholly owned subsidiary of Honeywell International Inc., for the U.S. Department of Energy’s National Nuclear Security Administration under contract DE-NA0003525.

This paper describes objective technical results and analysis. Any subjective views or opinions that might be expressed in the paper do not necessarily represent the views of the U.S. Department of Energy or the United States Government. 

\subsection*{AI Acknowledgment}

AI proposed proof strategies for \Cref{lem:part_lap} and \Cref{lem:worst_case_rank}. The overarching proof strategies and ideas present in this paper were human generated (with the exception of \Cref{ft:AI}).  AI tools were also used to find typos.

\bibliographystyle{plain}
\bibliography{citations.bib}

@article{bouchet1988graphic,
  title={Graphic presentations of isotropic systems},
  author={Bouchet, Andr{\'e}},
  journal={Journal of Combinatorial Theory, Series B},
  volume={45},
  number={1},
  pages={58--76},
  year={1988},
  publisher={Elsevier}
}

@article{grice2011arbitrarily,
  title={Arbitrarily complete Bell-state measurement using only linear optical elements},
  author={Grice, Warren P},
  journal={Physical Review A—Atomic, Molecular, and Optical Physics},
  volume={84},
  number={4},
  pages={042331},
  year={2011},
  publisher={APS}
}

@misc{silverman2004symmetric,
  author       = {Silverman, Jared S.},
  title        = {The Probability that a Random Symmetric Matrix 
                  over $\mathbb{F}_2$ is Nonsingular},
  howpublished = {REU report, Clemson University},
  year         = {2004},
  url          = {https://www.math.clemson.edu/~kevja/REU/2004/SymmetricRankRMatrices.pdf}
}

@article{ewert20143,
  title={3/4-efficient bell measurement with passive linear optics and unentangled ancillae},
  author={Ewert, Fabian and van Loock, Peter},
  journal={Physical review letters},
  volume={113},
  number={14},
  pages={140403},
  year={2014},
  publisher={APS}
}

@article{zaidi2013beating,
  title={Beating the one-half limit of ancilla-free linear optics Bell measurements},
  author={Zaidi, Hussain A and van Loock, Peter},
  journal={Physical review letters},
  volume={110},
  number={26},
  pages={260501},
  year={2013},
  publisher={APS}
}

@article{khesin2025universal,
  title={Universal graph representation of stabilizer codes},
  author={Khesin, Andrey Boris and Lu, Jonathan Z and Shor, Peter W},
  journal={PRX Quantum},
  volume={6},
  number={4},
  pages={040325},
  year={2025},
  publisher={APS}
}

@inproceedings{goodenough2024bipartite,
  title={Bipartite entanglement of noisy stabilizer states through the lens of stabilizer codes},
  author={Goodenough, Kenneth and Sajjad, Aqil and Kaur, Eneet and Guha, Saikat and Towsley, Don},
  booktitle={2024 IEEE International Symposium on Information Theory (ISIT)},
  pages={545--550},
  year={2024},
  organization={IEEE}
}

@article{bartolucci2023fusion,
  title={Fusion-based quantum computation},
  author={Bartolucci, Sara and Birchall, Patrick and Bombin, Hector and Cable, Hugo and Dawson, Chris and Gimeno-Segovia, Mercedes and Johnston, Eric and Kieling, Konrad and Nickerson, Naomi and Pant, Mihir and others},
  journal={Nature Communications},
  volume={14},
  number={1},
  pages={912},
  year={2023},
  publisher={Nature Publishing Group UK London}
}

@article{azuma2015all,
  title={All-photonic quantum repeaters},
  author={Azuma, Koji and Tamaki, Kiyoshi and Lo, Hoi-Kwong},
  journal={Nature communications},
  volume={6},
  number={1},
  pages={6787},
  year={2015},
  publisher={Nature Publishing Group UK London}
}

@article{li2015resource,
  title={Resource costs for fault-tolerant linear optical quantum computing},
  author={Li, Ying and Humphreys, Peter C and Mendoza, Gabriel J and Benjamin, Simon C},
  journal={Physical Review X},
  volume={5},
  number={4},
  pages={041007},
  year={2015},
  publisher={APS}
}

@article{lee2015nearly,
  title={Nearly deterministic Bell measurement for multiphoton qubits and its application to quantum information processing},
  author={Lee, Seung-Woo and Park, Kimin and Ralph, Timothy C and Jeong, Hyunseok},
  journal={Physical review letters},
  volume={114},
  number={11},
  pages={113603},
  year={2015},
  publisher={APS}
}

@article{reiss2026optimal,
  title={Optimal logical Bell measurements on stabilizer codes with linear optics},
  author={Rei{\ss}, Simon D and van Loock, Peter},
  journal={arXiv preprint arXiv:2601.08820},
  year={2026}
}

@article{schmidt2019efficiencies,
  title={Efficiencies of logical Bell measurements on Calderbank-Shor-Steane codes with static linear optics},
  author={Schmidt, Frank and van Loock, Peter},
  journal={Physical Review A},
  volume={99},
  number={6},
  pages={062308},
  year={2019},
  publisher={APS}
}

@article{laflamme1996perfect,
  title={Perfect quantum error correcting code},
  author={Laflamme, Raymond and Miquel, Cesar and Paz, Juan Pablo and Zurek, Wojciech Hubert},
  journal={Physical Review Letters},
  volume={77},
  number={1},
  pages={198},
  year={1996},
  publisher={APS}
}

@article{pettersson2025long,
  title={Long-distance quantum communication using concatenated ring graph codes},
  author={Pettersson, Love and S{\o}rensen, Anders S},
  journal={Physical Review Applied},
  volume={24},
  number={4},
  pages={044090},
  year={2025},
  publisher={APS}
}

@article{lee2019fundamental,
  title={Fundamental building block for all-optical scalable quantum networks},
  author={Lee, Seung-Woo and Ralph, Timothy C and Jeong, Hyunseok},
  journal={Physical Review A},
  volume={100},
  number={5},
  pages={052303},
  year={2019},
  publisher={APS}
}

@inproceedings{cross2008codeword,
  title={Codeword stabilized quantum codes},
  author={Cross, Andrew and Smith, Graeme and Smolin, John A and Zeng, Bei},
  booktitle={2008 IEEE International Symposium on Information Theory},
  pages={364--368},
  year={2008},
  organization={IEEE}
}

@article{shor1995scheme,
  title={Scheme for reducing decoherence in quantum computer memory},
  author={Shor, Peter W},
  journal={Physical review A},
  volume={52},
  number={4},
  pages={R2493},
  year={1995},
  publisher={APS}
}

@article{aasen2023measurement,
  title={Measurement quantum cellular automata and anomalies in Floquet codes},
  author={Aasen, David and Haah, Jeongwan and Li, Zhi and Mong, Roger SK},
  journal={arXiv preprint arXiv:2304.01277},
  year={2023}
}

@article{browne2005resource,
  title={Resource-efficient linear optical quantum computation},
  author={Browne, Daniel E and Rudolph, Terry},
  journal={Physical Review Letters},
  volume={95},
  number={1},
  pages={010501},
  year={2005},
  publisher={APS}
}

@article{ewert2016ultrafast,
  title={Ultrafast long-distance quantum communication with static linear optics},
  author={Ewert, Fabian and Bergmann, Marcel and Van Loock, Peter},
  journal={Physical review letters},
  volume={117},
  number={21},
  pages={210501},
  year={2016},
  publisher={APS}
}

@article{hilaire2023linear,
  title={Linear optical logical Bell state measurements with optimal loss-tolerance threshold},
  author={Hilaire, Paul and Castor, Yaron and Barnes, Edwin and Economou, Sophia E and Grosshans, Fr{\'e}d{\'e}ric},
  journal={PRX quantum},
  volume={4},
  number={4},
  pages={040322},
  year={2023},
  publisher={APS}
}

@article{bell2023optimizing,
  title={Optimizing graph codes for measurement-based loss tolerance},
  author={Bell, Thomas J and Pettersson, Love A and Paesani, Stefano},
  journal={PRX Quantum},
  volume={4},
  number={2},
  pages={020328},
  year={2023},
  publisher={APS}
}

@article{schlingemann2001quantum,
  title={Quantum error-correcting codes associated with graphs},
  author={Schlingemann, Dirk and Werner, Reinhard F},
  journal={Physical Review A},
  volume={65},
  number={1},
  pages={012308},
  year={2001},
  publisher={APS}
}

@book{nielsen2001quantum,
  title={Quantum computation and quantum information},
  author={Nielsen, Michael A and Chuang, Isaac L},
  volume={2},
  year={2001},
  publisher={Cambridge university press Cambridge}
}

@article{hein2006entanglement,
  title={Entanglement in graph states and its applications},
  author={Hein, Marc and D{\"u}r, Wolfgang and Eisert, Jens and Raussendorf, Robert and Nest, M and Briegel, H-J},
  journal={arXiv preprint quant-ph/0602096},
  year={2006}
}

@article{van2004graphical,
  title={Graphical description of the action of local Clifford transformations on graph states},
  author={Van den Nest, Maarten and Dehaene, Jeroen and De Moor, Bart},
  journal={Physical Review A},
  volume={69},
  number={2},
  pages={022316},
  year={2004},
  publisher={APS}
}

@article{kapshikar2023hardness,
  title={On the hardness of the minimum distance problem of quantum codes},
  author={Kapshikar, Upendra and Kundu, Srijita},
  journal={IEEE Transactions on Information Theory},
  volume={69},
  number={10},
  pages={6293--6302},
  year={2023},
  publisher={IEEE}
}

@inproceedings{javelle2012minimum,
  title={On the minimum degree up to local complementation: Bounds and complexity},
  author={Javelle, J{\'e}r{\^o}me and Mhalla, Mehdi and Perdrix, Simon},
  booktitle={International Workshop on Graph-Theoretic Concepts in Computer Science},
  pages={138--147},
  year={2012},
  organization={Springer}
}

@article{mhalla2012graph,
  title={Graph states, pivot minor, and universality of (X, Z)-measurements},
  author={Mhalla, Mehdi and Perdrix, Simon},
  journal={arXiv preprint arXiv:1202.6551},
  year={2012}
}

@inproceedings{hoyer2006resources,
  title={Resources required for preparing graph states},
  author={H{\o}yer, Peter and Mhalla, Mehdi and Perdrix, Simon},
  booktitle={International Symposium on Algorithms and Computation},
  pages={638--649},
  year={2006},
  organization={Springer}
}

@inproceedings{cattaneo2015minimum,
  title={Minimum degree up to local complementation: Bounds, parameterized complexity, and exact algorithms},
  author={Cattan{\'e}o, David and Perdrix, Simon},
  booktitle={International Symposium on Algorithms and Computation},
  pages={259--270},
  year={2015},
  organization={Springer}
}

@inproceedings{claudet2024covering,
  title={Covering a graph with minimal local sets},
  author={Claudet, Nathan and Perdrix, Simon},
  booktitle={International Workshop on Graph-Theoretic Concepts in Computer Science},
  pages={136--150},
  year={2024},
  organization={Springer}
}

@article{brijder2015isotropic,
  title={Isotropic matroids I: Multimatroids and neighborhoods},
  author={Brijder, Robert and Traldi, Lorenzo},
  journal={arXiv preprint arXiv:1503.04406},
  year={2015}
}

@article{oumazouz2025classification,
  title={On the classification of graphs induced by a sequence of local complementations of Paley graphs},
  author={Oumazouz, Zhour and others},
  journal={Journal of Combinatorial Mathematics and Combinatorial Computing},
  volume={126},
  pages={29--72},
  year={2025},
  publisher={Combinatorial Press}
}

@article{goodenough2024near,
  title={Near-term n to k distillation protocols using graph codes},
  author={Goodenough, Kenneth and De Bone, Sebastian and Addala, Vaishnavi and Krastanov, Stefan and Jansen, Sarah and Gijswijt, Dion and Elkouss, David},
  journal={IEEE Journal on Selected Areas in Communications},
  volume={42},
  number={7},
  pages={1830--1849},
  year={2024},
  publisher={IEEE}
}

@article{CHVATAL1979285,
title = {The tail of the hypergeometric distribution},
journal = {Discrete Mathematics},
volume = {25},
number = {3},
pages = {285-287},
year = {1979},
issn = {0012-365X},
doi = {https://doi.org/10.1016/0012-365X(79)90084-0},
url = {https://www.sciencedirect.com/science/article/pii/0012365X79900840},
author = {V. Chvátal}
}

\appendix

\section{Proofs of Lemma \ref{lemma:dimension} and \ref{lemma:bouchet}}\label{section:proof_lemma}

\begin{lemma}
Let $\mathcal{C}$ be an $[[n, k, d]]$ stabilizer code with stabilizer group $\mathcal{S}$, and let $\mathcal{S}^\perp$ be its dual. Let $Q$ be a full-weight Pauli string on $n$ qubits. Then 

\begin{align}
\dim\left(\hat{Q}\cap \mathcal{S}^\perp\right) - \dim\left(\hat{Q}\cap \mathcal{S}\right) = k ,
\end{align}
where here we interpret $\hat{Q}\cap \mathcal{S}^\perp$ and $\hat{Q}\cap \mathcal{S}$ as subspaces of $\hat{Q}\cong\mathbb{F}_2^n$.
\end{lemma}

\begin{proof}
Define $\hat{Q} + \mathcal{S}:=\textrm{span}
(\hat{Q} \cup \mathcal{S})$ to be the set of all elements of the form $a+b$, with $a\in \hat{Q}$ and $b\in \mathcal{S}$, where we used $+$ for the binary operation since we are modding out phases. Now note that

\begin{align}
\left(\hat{Q} + \mathcal{S}\right)^\perp =\hat{Q}^\perp \cap \mathcal{S}^\perp =  \hat{Q}\cap \mathcal{S}^\perp\ ,\label{eq:push_perp}
\end{align}
where in the first equality we used that an element commutes with all of $\hat{Q} + \mathcal{S}$ if and only if it commutes with all of $\hat{Q}$ and $\mathcal{S}$. In the second equality we used that $\hat{Q}^\perp = \hat{Q}$, since $\hat{Q}$ is the stabilizer group of a state. From this we find

\begin{align}
\dim(\hat{Q}\cap \mathcal{S}^\perp) = & \dim((\hat{Q}+\mathcal{S})^\perp)\\
=&\,2n - \dim(\hat{Q}+\mathcal{S})\\
=&\,2n-\left(\dim (\hat{Q})+\dim(\mathcal{S})-\dim(\hat{Q}\cap \mathcal{S})\right)\\
=&\,2n-(n+(n-k)-\dim(\hat{Q}\cap \mathcal{S}))\\
=&\,k+\dim(\hat{Q}\cap \mathcal{S}) .\label{eq:final_step_lemma}
\end{align}
In the first equality we used Eq.~\eqref{eq:push_perp}. In the second equality we used that for any subspace $W$ in a symplectic vector space $\mathcal{W}$ it holds that $\dim(\mathcal{W}) = \dim(W) + \dim(W^\perp)$ \cite{nielsen2001quantum}. The third equality is the dimension formula for the sum of vector spaces, and the last equality follows from basic rewriting. The lemma now follows from rewriting Eq.~\eqref{eq:final_step_lemma}.
\end{proof}

\begin{lemma}
Let $\mathcal{C}$ be an $[[n, k, d]]$ stabilizer code with stabilizer group $\mathcal{S}$, and let $P$ be any (not necessarily full-weight) Pauli string on $n$ qubits. If $\dim(\hat{P}\cap \mathcal{S})=m$, then there exists a full-weight superstring $Q$ of $P$ such that $\dim(\hat{Q}\cap \mathcal{S})=m$.
\end{lemma}
\begin{proof}
We may assume that $P$ is not full-weight, since otherwise the statement is clearly true. Since $P$ is not full-weight it has an identity entry at some index $i$. We now set $P_X = P\cdot X_i$, $P_Y=P\cdot Y_i$ and $P_Z=P\cdot Z_i$, and claim that there are always at least two choices of an element $P'$ from $\lbrace{P_X, P_Y, P_Z\rbrace}$ such that $\dim(\hat{P'}\cap \mathcal{S})=\dim(\hat{P}\cap \mathcal{S})=m$. Indeed, if this is true then it is always possible to continue extending $P$ until a full-weight $Q$ has been reached, such that $\dim(\hat{P}\cap \mathcal{S})=\dim(\hat{Q}\cap \mathcal{S})=m$.

\medskip

It thus remains to prove the claim that at least two of the $\lbrace{P_X, P_Y, P_Z\rbrace}$ satisfy $\dim(\hat{P'}\cap \mathcal{S})=\dim(\hat{P}\cap \mathcal{S})$. Assume, towards a proof by contradiction, that at least two distinct $P_A, P_B\in \lbrace{P_X, P_Y, P_Z\rbrace} $ satisfy $\dim(\hat{P_A}\cap \mathcal{S})\neq m$ and $\dim(\hat{P_B}\cap \mathcal{S})\neq m$. In particular, this means that $\dim(\hat{P_A}\cap \mathcal{S}) = \dim(\hat{P_B}\cap \mathcal{S}) = m+1$, by the following:  $\hat{P}_A$ is the disjoint union $\hat{P} \sqcup (\hat{P}A_i)$ where $\hat{P}A_i$ is the set of Pauli operators with $A$ at position $i$ and an element of $\hat{P}$ at the remaining positions.  It follows that $\dim(\hat{P}_A \cap \mathcal{S}) \leq m+1$ and similarly for $\dim(\hat{P}_B \cap\mathcal{S})$ since any two elements of $\hat{P}A_i\cap \mathcal{S}$ must be related by an element of $\hat{P}\cap \mathcal{S}$. So, starting with a basis of $\hat{P}\cap \mathcal{S}$ there can be at most one independent element of $\hat{P}A_i \cap \mathcal{S}$.

This means that both $(\hat{P_A}\cap \mathcal{S})\setminus \hat{P}$ and $(\hat{P_B}\cap \mathcal{S})\setminus \hat{P}$ are non-empty. As such, let $S_A$ and $S_B$ be two elements in $(\hat{P_A}\cap \mathcal{S})\setminus \hat{P}$ and $(\hat{P_B}\cap \mathcal{S})\setminus \hat{P}$, respectively. Note that $S_A$ and $S_B$ have distinct non-identity elements on the $i$'th entry, i.e.~$S_A$ and $S_B$ anti-commute at the $i$'th entry. But since $S_A$ and $S_B$ are stabilizers, they need to anti-commute an even number of times across all entries. However, $S_A$, $S_B$ commute at each entry $j$ in the support of $P$, since at each such entry $j$ we have that $P_A$ and $P_B$ are either identity, or are equal to $P(j)$. As such, $S_A$ and $S_B$ cannot commute, leading to a contradiction.
\end{proof}

\begin{lemma}[\cite{brijder2015isotropic}]
A stabilizer $S$ of a stabilizer state is called a \emph{minimal stabilizer} if $S$ is not the identity, and no proper substring of $S$ is a non-identity stabilizer. Let $\ket{\phi}$ be any stabilizer state, let $S$ be any of its minimal stabilizers, and fix an arbitrary $v\in \textrm{supp}(S)$. Then $\ket{\phi}$ is locally equivalent to a graph state $\ket{G}$ such that vertex $v$ has closed neighborhood equal to $\textrm{supp}(S)$.
\end{lemma}
\begin{proof}

It is well known that quantum stabilizer states can be converted to quantum graph states via local Clifford operations \cite{hein2006entanglement}.  Conjugating a minimal stabilizer by a local Clifford operation will yield another minimal stabilizer of the transformed state.  So, the lemma reduces to showing that for any minimal stabilizer of a graph state, we can find a local Clifford operation which maps the graph state to another graph state, and at the same time conjugates the minimal stabilizer to a canonical generator of the graph state (i.e.,~a stabilizer of the form $X_i \prod_{j\in N(i)}Z_j$).  We will give a sequence of local complements and corresponding local Clifford operations which accomplishes this.  

A local complement of the graph at a vertex $w\in [n]$ is equivalent to the Clifford operation $U_w=\sqrt{-iX_w}\prod_{u\in N(w)}\sqrt{iZ_u}$.  Let us note the following behavior with respect to conjugation: 
\begin{align}\label{eq:conj_relations}
\sqrt{-iX} X \sqrt{-iX}^\dagger&=X,\\
\nonumber \sqrt{-iX} Y \sqrt{-iX}^\dagger&=Z,\\
\nonumber \sqrt{-iX} Z \sqrt{-iX}^\dagger&=-Y,\\
\nonumber \sqrt{iZ} X \sqrt{iZ}^\dagger&=-Y,\\
\nonumber \sqrt{iZ} Y \sqrt{iZ}^\dagger&=X,\\
\nonumber \sqrt{iZ} Z \sqrt{iZ}^\dagger&=Z.
\end{align}
We will exhibit a sequence of minimal stabilizers $\{S^{(0)}, S^{(1)}, ...S^{(k)}\}$ such that $S^{(i+1)}$ is obtained from $S^{(i)}$ by conjugation with some Clifford operation composed of unitaries of the form $U_w$.  Let us define $A_\sigma^{(i)} = \{j\in [n]: S^{(i)}(j)=\sigma\}\subseteq [n]$.  The sequence of minimal stabilizers will have the property that $S^{(k)}$ has $A_X^{(k)}=\{w\}$ and $A_Y^{(k)}=\emptyset$ so that $A_Z^{(k)}=\textrm{supp}(S^{(k)})\setminus \{w\}$.  Then an edge pivot \cite{mhalla2012graph} along the edge $(v, w)$ (which corresponds to $U_v U_w U_v$) converts $S^{(k)}$ to a canonical graph state generator as defined in the lemma, while ensuring that the final state is a graph state.

First, while $A_Y^{(i)}$ is nonempty we will choose an element $w \in A_Y^{(i)}$, and then map $S^{(i)}$ to $S^{(i+1)}:=U_w S^{(i)}U_w^\dagger$.  By \Cref{eq:conj_relations}, (1) $A_Z^{(i+1)}=A_Z^{(i)} \cup \{w\}$ and (2) $\left(A_X^{(i)} \cup A_Y^{(i)}\right)\setminus \{w\}=A_X^{(i+1)} \cup A_Y^{(i+1)}$.  So this has the effect of moving $w$ into $A_Z^{(i+1)}$ while shuffling elements of $A_X^{(i)}$ and $A_Y^{(i)}$ among each other. Proceeding in this way we can strictly increase the size of $A_Z^{(i)}$.

After complementing on available $Y$ operators we will arrive at some $S^{(j)}$ with $A_Y^{(j)}$ empty, but possibly with $|A_X^{(j)}| \geq 1$ (we cannot arrive at a minimal stabilizer with $A_X^{(j)}=\emptyset=A_Y^{(j)}$, since every non-identity stabilizer of a graph state contains at least one $X$ or $Y$ operator). Our goal will be to turn all but one of these $X$ operators into $Y$ operators, such that they in turn can be transformed into $Z$ operators. If $|A_X^{(j)}| =1 $ we are left with a canonical generator, so assume $|A_X^{(j)}| > 1$. Consider first the case where there are edges inside $G[A_X^{(j)}]$ (where $G$ is the graph after the local complementations). In this case, we can locally complement on a vertex incident to such an edge to produce a $Y$ operator, while keeping the $X$ operator on the vertex we complement on the same. Second, consider the case where $G[{A_X}^{(j)}]$ does not contain any edges, and there is some vertex $w \in A_\mathbb{I}^{(j)}$ which is adjacent to a vertex $u$ in $A_{X}^{(j)}$. In this case, implementing $U_w$ maps the $X$ operator on $u$ to $Y$.  

Finally, if $G[A_X^{(j)}]$ is edgeless and there is no $w \in A_\mathbb{I}^{(j)}$ that is adjacent to a $u \in A_X^{(j)}$, then all neighbors of any vertex $u\in A_X^{(j)}$ are contained in $A_{Z}^{(j)}$. This follows since $A_Y^{(j)}=\emptyset$, there are no edges in $G[A_X^{(j)}]$, and by assumption there are no neighbors in $A_\mathbb{I}^{(j)}$. However, this cannot occur since $X_uZ_{N(u)}$ is always a stabilizer of the corresponding graph state, so that $S^{(j)}$ would be a proper superstring of $X_uZ_{N(u)}$. The fact that it is a proper superstring follows from the assumption that $\left|A_X^{(j)}\right|>1$, while $X_uZ_{N(u)}$ has only one $X$ operator. This contradicts the assumption that $S^{(j)}$ is a minimal stabilizer.

As such, whenever $\left|A_{X}^{(j)}\right|>1$ it is always possible to strictly increase the size of the set $A_Z^{(j)}$, so that eventually we arrive at some $S^{(k)}$ satisfying $A_X^{(k)}=\{w\}$ and $A_Y^{(k)}=\emptyset$. We then have that $A_Z^{(k)}=\textrm{supp}(S^{(k)})\setminus \{w\}$, so that an edge pivot $U_v U_wU_v$ moves the support of the $X$ operator to qubit $v$, finishing the lemma.  

\end{proof}

\section{Most graph codes admit perfect strategies}\label{sec:gen_proof}

\subsection{Reduction to Partial Laplacian}\label{sec:part_lap}
Let $G$ be a random graph and let $\ket{G}$ be the corresponding graph state. Our goal will be to show that any of the $[[n, 1,d]]$ CWS codes associated to $\ket{G}$ will---with high probability---admit a perfect fusion strategy. 

From the discussion in Section \ref{sec:graph_codes}, this is equivalent to showing that for $G$ there exists w.h.p.~some $S$ such that $\Gamma_S:=\prod_{s\in S}\Gamma_s$ is full weight and that $\Gamma_S$ does not contain any proper non-trivial stabilizers. Note that any substring of $\Gamma_S$ that is also a stabilizer is necessarily of the form $\Gamma_R$ for some $R\subseteq S$. This can be seen as follows. First, every stabilizer is of the form $\Gamma_R$ for some arbitrary $R$. Second, $v\in R$ iff $\Gamma_R[v]\in \lbrace{ X, Y\rbrace}$, which means that if there exists a $v\in R$ that is not in $S$, then $\Gamma_R[v]\neq\Gamma_S[v]$ yet neither $\Gamma_R[v]$ nor $\Gamma_S[v]$ equals $I$.

Let us consider the subgraph induced by the vertices in $S$ and define $L_S\in \mathbb{F}_2^{|S| \times |S|}$ to be the Laplacian of the induced subgraph.  In this context a Laplacian has off-diagonal entries matching the adjacency matrix and the diagonal entry at vertex $v$ equals $\left|N(v)\right| \textrm{ mod } 2$.

\begin{lemma}\label{lem:part_lap}
    Let $\Gamma_S$ be full weight.  Then, $\Gamma_S$ has a proper non-trivial $\Gamma_R$ as a substring iff $\text{rank}(L_S) < |S|-1$.  
\end{lemma}
\begin{proof}
    First note that $L_S \mathds{1}_{|S|}=0$ so $L_S$ is never full rank. For any vector $v\in \mathbb{F}_2^{|S|}$, $(L_S  v)_i$ is the number of edges adjacent to $i$ in the induced subgraph which are in the cut naturally defined by $v$: If $v_i=1$ then $(L_S  v)_i=|\{j\in N(i): v_j=1\}|+|N(v)|=|\{j \in N(i): v_j=0\}|$ whereas if $v_i=0$ then $(L_S  v)_i=|\{j\in N(i): v_j=1\}|$.

    Now suppose that $\Gamma_S$ has a sub-stabilizer $\Gamma_R$ for $R\subset S$.  This implies that $\Gamma_{S\setminus R}$ is also a substring and disjoint from $\Gamma_R$.  This implies that the image of $L_S \mathds{1}_R$ restricted to $S\setminus R$ is $0$ and that the image of $L_S \mathds{1}_{S\setminus R}$ restricted to $R$ is $0$.  These two together imply that when we consider the cut in $G[S]$ defined by $R$ each vertex in $S$ has an even number of neighbors crossing the cut.  This in turn implies that $L_S\mathds{1}_R=0$.

    Now suppose that $L_S$ has some null vector defined by $\mathds{1}_R$ for $R\subset S$. Then $\mathds{1}_{S\setminus R}$ is also a null vector.  This implies that the action of $\Gamma_R$ on the vertices in $S\setminus R$ is $\mathbb{I}$ and that the action of $\Gamma_R$ matches $\Gamma_S$ on $R$.  The action of $\Gamma_R$ on vertices in $V\setminus S$ is either $\mathbb{I}$ or $Z$.  Since $\Gamma_S$ is full weight, $\Gamma_R$ is a substring.  
    
\end{proof}

\begin{lemma}\label{lem:sub_laplace}
    Let $G$ be a graph and $L$ be the corresponding Laplacian.  Let $L'$ be the submatrix obtained by deleting the first row and column from $L$.  Then $\text{rank}(L)=n-1$ iff $\text{rank}(L')=n-1$. 
\end{lemma}
\begin{proof}
    As noted before, since $L$ is a Laplacian $\mathds{1}$ is always a null vector.  Hence if $L$ has a null vector $v$ which is distinct from $\mathds{1}$ then we can always find a null vector distinct from $\mathds{1}$ with a $0$ in the first entry.  This vector is automatically a null vector for $L'$.  This gives the $\Leftarrow$ direction.  

    Now let us assume that the matrix $L'$ has a null vector $v$.  The matrix $L$ is of the form $\begin{bmatrix}|w| & w^T\\w &  L' \end{bmatrix}$ for some vector $w$.  If $v\cdot{} w=0$ then $[0, v]$ is a null vector for $L$.  If $v\cdot{} w=1$ then by the symmetry of $L$: $0=\mathds{1}^T L \begin{bmatrix} 0 \\ v\end{bmatrix}=\mathds{1}^T \begin{bmatrix} w\cdot{} v \\0\end{bmatrix}=1$ so it cannot be true that $v\cdot{} w=1$.
    
\end{proof}

\subsection{Analysis of Random variables}

\begin{lemma}\cite[Theorem 0.1]{silverman2004symmetric}\label{lem:rand_sym}
    Let $A$ be a uniform random symmetric matrix with entries in $\mathbb{F}_2$.  Then the probability that $A$ is full rank is:
    \begin{equation}
        q_n=\prod_{j=1}^{\lceil n/2 \rceil} (1-2^{-(2 j-1)}) .
    \end{equation}
\end{lemma}

We will use the convention $q_t=1$ for $t \leq 0$.

\begin{lemma}\label{lem:q_convergence}
    Define 
    $$
    q_{\infty}=\prod_{j=1}^{\infty} (1-2^{-(2 j-1)})\approx 0.42.
    $$ 
    Then, for $c \in (0, 1)$, 
    $$q_{\lfloor cn \rfloor} = q_{\infty}\bc{1+\Theta(2^{-cn})}.$$
\end{lemma}
\begin{proof}
    Consider the ratio 
    \begin{align*}
    \frac{q_{\lfloor cn\rfloor}}{q_{\infty}} &= \frac{\prod_{j=1}^{\lceil \lfloor cn\rfloor/2 \rceil} (1-2^{-(2 j-1)})}{\prod_{j=1}^{\infty} (1-2^{-(2 j-1)})}\\
    &=\prod_{j=\lceil \lfloor cn\rfloor/2 \rceil+1} ^{\infty}\bc{1-2^{-(2 j-1)}}^{-1}.
    \end{align*}
    Taking the log and using Taylor expansion of log, we get
    \begin{align*}
    \log\frac{q_{\lfloor cn\rfloor}}{q_{\infty}} &= \sum_{j=\lceil \lfloor cn\rfloor/2 \rceil+1}^\infty - \log\bc{1-2^{-(2j-1)}} \\
    &= \sum_{j=\lceil \lfloor cn\rfloor/2 \rceil+1}^\infty \bc{2^{-(2j-1)} + O(2^{-2(2j-1)})} \\
    &= \Theta(2^{-cn}) + O(2^{-2cn}) = \Theta(2^{-cn}).\end{align*}
    Finally, 
    \[\frac{q_{\lfloor cn\rfloor}}{q_{\infty}} = e^{\Theta(2^{-cn})} = 1+\Theta(2^{-cn}).\]
\end{proof}

\begin{lemma}\label{lem:worst_case_rank}
    Let $M=\begin{bmatrix} A & B\\ B^T & C\end{bmatrix}$ be a symmetric matrix where $B\in \mathbb{F}_2^{m\times p}$ and $C\in\mathbb{F}_2^{p\times p}$ are fixed but $A\in \mathbb{F}_2^{m\times m}$ is uniformly sampled over all symmetric matrices.  Then, $\mathbb{P}[M \text{ is full rank}] \leq q_{m-\text{corank}(C)}$ .
\end{lemma}
\begin{proof}
    Let $r=\text{rank}(C)$.  $M$ is full rank iff $\begin{bmatrix} \mathbb{I} & 0 \\ 0 & Q\end{bmatrix} M \begin{bmatrix} \mathbb{I} & 0 \\ 0 & Q^T\end{bmatrix}$ is full rank for any full rank $Q$.  There exists a matrix $Q$ such that $Q C Q^T=\begin{bmatrix} 0 & 0\\ 0 & C' \end{bmatrix}$ where $C'$ is a full rank $r\times r$ matrix.  So instead of $M$ we can consider $\begin{bmatrix} A & B_1 & B_2 \\B_1^T & 0 & 0\\ B_2^T  & 0 & C' \end{bmatrix}$.  $C'$ is full rank, so we can take the Schur complement on that block to infer that $M$ is full rank iff $M'=\begin{bmatrix} A +B_2 (C')^{-1} B_2^T & B_1 \\ B_1^T & 0\end{bmatrix}$ is full rank.  $B_2 (C')^{-1} B_2^T$ is some fixed symmetric matrix and $A$ is a random uniform symmetric matrix so the probability that $M'$ is full rank is the same as the probability that $M''=\begin{bmatrix}A & B_1 \\ B_1^T & 0 \end{bmatrix}$ is full rank for a uniform random symmetric $A$.  If $B_1$ does not have full column rank then the probability that $M''$ is full rank is zero and the bound holds trivially.  Otherwise, by the standard reduced row-echelon form we can find invertible matrices $W_1$ and $W_2$ such that $W_1 B_1 W_2^T=\begin{bmatrix}\mathbb{I} \\0 \end{bmatrix}$.  The matrix $M''$ is full rank iff the matrix $M'''=\begin{bmatrix} W_1 & 0 \\ 0 & W_2 \end{bmatrix} M''\begin{bmatrix} W_1^T & 0 \\ 0 & W_2^T \end{bmatrix}=\begin{bmatrix} W_1 A W_1^T & \begin{matrix} \mathbb{I} \\ 0\end{matrix}\\\begin{matrix} \mathbb{I} & 0\end{matrix} & 0\end{bmatrix}$ is full rank.  Again the distribution of $W_1 A W_1^T$ is the same distribution as $A$ so $M'''$ is full rank iff $\begin{bmatrix}  A  & \begin{matrix} \mathbb{I} \\ 0\end{matrix}\\\begin{matrix} \mathbb{I} & 0\end{matrix} & 0\end{bmatrix}$ is full rank.  We can then use a Schur complement to finish the proof.  
\end{proof}

\begin{lemma}\label{lem:d_d_transpose}
    Let $D\in \mathbb{F}_2^{m \times p}$ be uniformly random.  Let $v\in \mathbb{F}_2^m$ and $w\in \mathbb{F}_2^p$ be fixed.  Then,
    \begin{align}
        \mathbb{P}[D\mathds{1}=v, w=D^T \mathds{1}]=\begin{cases}
            2^{-(m+p-1)} \text{ if } |w|\equiv|v| \pmod{2},\\
            0 \text{  o.w.  }
        \end{cases}
    \end{align}
\end{lemma}

\begin{proof}
    Let $r_i(D) = \sum_j D_{ij}$ and $c_j(D) = \sum_i D_{ij}$ denote the $i$th row sum and the $j$th column sum modulo 2, respectively.  The event is the system of $m+p$ affine equations $r_i(D)=v_i, c_j(D)=w_j$ in the $mp$ entries of $D$.  
    Since over $\mathbb{F}_2$, $\sum_i r_i(D)=\sum_j c_j(D)$ identically, the system is inconsistent unless $|v|\equiv|w| \pmod 2$.  If this congruence fails, the probability is 0.
    Thus we now assume $|v|\equiv|w|$.  To compute the rank of these $m+p$ equations, we ask if there is a nontrivial linear combination of the row-sum equations and column-sum equations that is identically zero.  Thus, consider
    \[\sum_i a_ir_i + \sum_j b_jc_j = 0, \quad a_i,b_j\in\mathbb{F}_2.\]  This equation means that for any $D$, this linear combination is zero.  We expand the equation with the definitions of $r_i(D)$ and $c_j(D)$:
    \begin{align*}
    \sum_ia_ir_i(D) + \sum_jb_jc_j(D) &= \sum_ia_i\sum_jD_{ij} + \sum_jb_j\sum_iD_{ij}\\
    &= \sum_{i,j}a_iD_{ij} + \sum_{i,j}b_jD_{ij} = \sum_{i,j}(a_i+b_j)D_{ij}.
    \end{align*}  Since this identity holds for every $D$, every coefficient is 0, so $a_i+b_j=0$.  This implies $a_i=b_j$ for every $i,j$.  Thus all $a_i$ and all $b_j$ are equal to a common element of $\mathbb{F}_2$.  The space of linear dependencies among the $m+p$ equations is one-dimensional and therefore the rank is $m+p-1$.
    Since the parity condition holds, the affine system is consistent.
    Therefore its solution set is an affine subspace of
    $\mathbb F_2^{m\times p}$ of dimension
    \[
        mp-(m+p-1).
    \]
    Since $D$ is uniform over $2^{mp}$ matrices, the probability of the event is
    \[
        \frac{2^{mp-(m+p-1)}}{2^{mp}}
        =
        2^{-(m+p-1)}.
    \]
\end{proof}
\begin{lemma}\label{lem:corank_small}
    Let $A\in \mathbb{F}_2^{m\times m}$ be a uniform random symmetric matrix.  For any $\delta \in (0, 1)$ there exists some constant $c$ independent of $m$ such that $\mathbb{P} [corank(A) \geq \delta m ] \leq 2^{-c m^2}$.
\end{lemma}  
\begin{proof}
If $\corank(A)\geq k$ then $\ker A$ contains some $k$-dimensional subspace. Thus,
\[
    \cbc{\corank(A)\geq k} \subseteq
    \bigcup_{\substack{V\subseteq \FF_2^m \\ \dim V=k}} \cbc{V\subseteq \ker A},
\]
and a union bound yields
\[
    \PP\brk{\corank(A)\geq k}
    \leq \sum_{\substack{V\subseteq \FF_2^m \\ \dim V=k}} \PP\brk{V\subseteq \ker A}.
\]
We now show that the summand depends only on $\dim V$, so that it suffices to evaluate it for a single convenient $V$. This follows from the invariance of the distribution of $A$ under congruence. 
For invertible $Q\in\FF_2^{m\times m}$ set
$\tau_Q(A)=QAQ^{T}$. Then $\tau_Q$ preserves symmetry, since
$(QAQ^T)^T=QA^TQ^T=QAQ^T$, and $\tau_{Q_1}\circ\tau_{Q_2}=\tau_{Q_1Q_2}$, so
$Q\mapsto\tau_Q$ is a homomorphism on $\GL_m(\FF_2)$; in particular
$\tau_Q\circ\tau_{Q^{-1}}=\tau_I=\id$, so each $\tau_Q$ is a bijection of the set of
symmetric $m\times m$ matrices over $\FF_2$. A bijection of a finite set carries the
uniform distribution to itself, so $\tau_Q(A)$ has the same law as $A$.
The action interacts with kernels in the expected way: $\ker \tau_Q(A)=Q^{-T}\ker A$,
where $Q^{-T}:=(Q^T)^{-1}$. Indeed, for $x\in\FF_2^m$,
\[
    QAQ^{T}x=0
    \iff A\bc{Q^{T}x}=0
    \iff Q^{T}x\in\ker A
    \iff x\in Q^{-T}\ker A,
\]
where the first equivalence is given by invertibility of $Q$.
Let $V,W$ be any two $k$-dimensional subspaces of $\FF_2^m$. Choose an invertible
$R$ with $RV=W$ (extend bases of $V$ and $W$ to bases of $\FF_2^m$) and put $Q:=R^{-T}$,
so that $Q^{-T}=R$. Since $R$ is injective,
\[
    V\subseteq\ker A
    \iff RV\subseteq R\ker A
    \iff W\subseteq Q^{-T}\ker A
    \iff W\subseteq \ker \tau_Q(A),
\]
and therefore
\[
    \PP\brk{V\subseteq\ker A}
    =\PP\brk{W\subseteq \ker\tau_Q(A)}
    =\PP\brk{W\subseteq\ker A},
\]
where the last equality follows by the distributional invariance just established. Note that only
the subspace is being normalized here. $A$ itself remains uniform, so no canonical
form for symmetric matrices over $\FF_2$ is required.
We may thus take $V=\textrm{span}(e_1,\dots,e_k)$, for which $V\subseteq\ker A$ says exactly
that the first $k$ columns of $A$ vanish. A symmetric matrix is determined by its
free entries $\cbc{A_{ij}: i\leq j}$, which are i.i.d.\ uniform bits, and those forced
to zero are the entries with $i\leq j$ and $i\leq k$, of which there are
$\sum_{i=1}^{k}(m-i+1)=km-\binom{k}{2}$. Hence
\[
    \PP\brk{V\subseteq\ker A}=2^{-\bc{km-\binom{k}{2}}}.
\]
Finally, the number of $k$-dimensional subspaces of $\FF_2^m$ is the Gaussian
binomial coefficient
\[
    \binom{m}{k}_{2}=\prod_{i=0}^{k-1}\frac{2^m-2^i}{2^k-2^i}
    \leq C\cdot 2^{k(m-k)},
    \qquad C=\prod_{i\geq 1}\bc{1-2^{-i}}^{-1}<4 ,
\]
so that
\[
    \PP\brk{\corank(A)\geq k}\leq 4\cdot 2^{k(m-k)}\cdot 2^{-km+\binom{k}{2}}
    = 4\cdot 2^{-k^2+\frac{k(k-1)}{2}}
    = 4\cdot 2^{-k(k+1)/2}.
\]
Taking $k=\lceil \delta m\rceil\geq \delta m$ gives $k(k+1)/2\geq \delta^2m^2/2$, so
$\PP\brk{\corank(A)\geq\delta m}\leq 4\cdot 2^{-\delta^2m^2/2}$, which is at most
$2^{-cm^2}$ for
any $c<\delta^2/2$ 
once $m$ is large.
\end{proof}

\begin{theorem}\label{thm:gen_perfect_technical}
    For each $S\subseteq V$ let $\mathbf{L}_S\in \mathbb{F}_2^{|S|\times |S|}$ be the Laplacian of the subgraph induced by $S$ and let $\mathbf{C}_S \in \mathbb{F}_2^{(n-|S|) \times |S|}$ be the cut matrix associated to the subset $S$.  Let $s=\lfloor n/2 \rfloor$, let $G$ be an $n$-vertex graph sampled uniformly and define 
    \begin{equation}
        \mathbf{N}=\sum_{S:|S|=s} \mathbf{1}\left\{\textrm{rank}(\mathbf{L}_S)=s-1, \mathds{1}_{V\setminus S}=\mathbf{C}_S \mathds{1}_S\right\}.
    \end{equation}
    Then,
    \begin{align}
        \mathbb{P}[\mathbf{N} \geq 1] \geq 1-O(\sqrt{n}2^{- n/144}).
    \end{align}
\end{theorem}
\begin{proof}
    For each $S$, $\mathbf{L}_S$ and $\mathbf{C}_S$ are independent.  By \Cref{lem:sub_laplace}, $rank(\mathbf{L}_S)=s-1$ iff the submatrix obtained by deleting the first row and column is full rank.  This submatrix is a uniformly random symmetric matrix so it is full rank with probability $q_{s-1}$ by \Cref{lem:rand_sym}.  $\mathds{1}_{V\setminus S}=\mathbf{C}_S \mathds{1}_S$ with probability $2^{-(n-s)}$, so 
    $$
\mathbb{P}[\textrm{rank}(\mathbf{L}_S)=|S|-1, \mathds{1}_{V\setminus S}=\mathbf{C}_S \mathds{1}_S]=2^{-(n-s) }q_{s-1}.  
    $$
    
    It follows that 
    $$
\mathbb{E}[\mathbf{N}]=\binom{n}{s}2^{-(n-s) }q_{s-1}=\binom{n}{s}2^{-\lceil n/2\rceil }q_{s-1}.
    $$

    Now we will demonstrate an estimate on $\mathbb{E}[\mathbf{N}^2]$ with the goal of using the second moment method to complete the proof.  Let us fix $S$ and $R$ and let their intersection be $I=S\cap R$.  We will first focus on the ``typical'' case when $|I|\in[n/4-\delta n, n/4+\delta n]$ for $\delta=1/24$. We will need to define a number of events.  Let,
    \begin{align}
        \textrm{Cut}_S&=\{\mathds{1}_{V\setminus S}=\mathbf{C}_S \mathds{1}_S\},\\
        \textrm{Cut}_R&=\{\mathds{1}_{V\setminus R}=\mathbf{C}_R \mathds{1}_R\},\\
        \textrm{Lap}_S&=\{\textrm{rank}(\mathbf{L}_S)=|S|-1\},\\
        \textrm{Lap}_R&=\{\textrm{rank}(\mathbf{L}_R)=|R|-1\},\\
        \textrm{Good}&=\textrm{Cut}_S\cap \textrm{Cut}_R \cap \textrm{Lap}_S \cap \textrm{Lap}_R.
    \end{align}

    Our goal is to find an upper bound on $\mathbb{P}[\textrm{Good}]$.  By symmetry, we can assume $S=\{1, 2, ..., s\}$ and $R=\{s-|I|+1, ..., n-|I|\}$.  Let $\mathbf{Adj}$ be the adjacency matrix of the graph $G$.  We need to define the following submatrices:

    \begin{align}
        \mathbf{A}&=\mathbf{L}_S[S\setminus\{1, I\}],\\
        \mathbf{B}&=\mathbf{L}_S[S\setminus\{1, I\}, I],\\
        \mathbf{C}&=\mathbf{L}_S[I],\\
        \mathbf{C}'&=\mathbf{L}_R[I],\\
        \mathbf{E}&=\mathbf{L}_R[I, R\setminus\{I, n-|I|\}],\\
        \mathbf{F}&=\mathbf{L}_R[R\setminus\{I, n-|I|\}],\\
        \mathbf{D}&=\mathbf{Adj}[S\setminus I, R\setminus I],\\
        \mathbf{H}_1&=\mathbf{Adj}[[n]\setminus (S\cup R), S\setminus I],\\
        \mathbf{H}_2&=\mathbf{Adj}[[n]\setminus (S\cup R), I],\\
        \mathbf{H}_3&=\mathbf{Adj}[[n]\setminus (S\cup R), R\setminus I].
    \end{align}

\begin{figure}
\centerfloat
\includegraphics[width=4in]{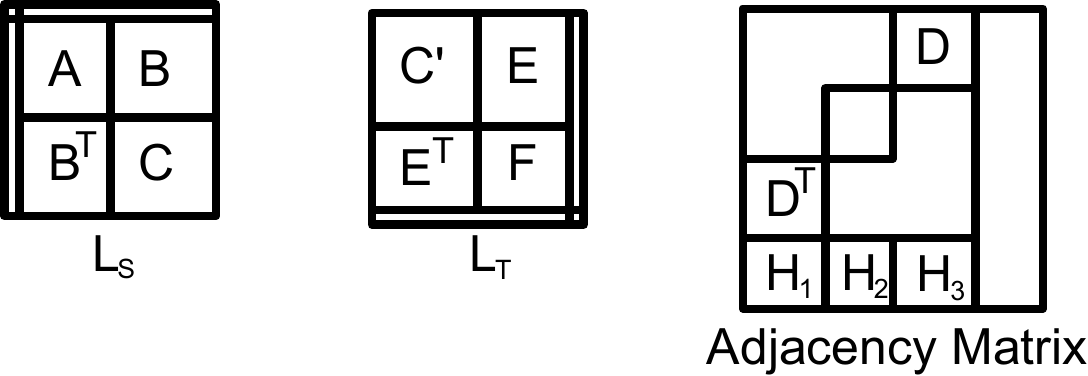}
\caption{Relevant submatrices for the proof of \Cref{thm:gen_perfect_technical}.}\label{fig:submatrices}
\end{figure}

 We can illustrate these submatrices with a picture (\Cref{fig:submatrices}). The matrices $\mathbf{C}$ and $\mathbf{C}'$ are the submatrices of $\mathbf{L}_S, \mathbf{L}_R$ corresponding to the intersection $I$.  They have the same off-diagonal elements but potentially distinct diagonal elements.  $G$ is random so the marginal distributions of $\begin{bmatrix} \mathbf{A} & \mathbf{B} \\ \mathbf{B}^T & \mathbf{C}\end{bmatrix}$ and $\begin{bmatrix} \mathbf{C}' & \mathbf{E} \\ \mathbf{E}^T & \mathbf{F} \end{bmatrix}$ are both uniform over the set of symmetric matrices since we have excluded the first row and column from each Laplacian.  

We want an upper bound on $\mathbb{E}[\mathbf{N}^2]$, which will require an upper bound on $\mathbb{P}[\textrm{Good}]$.  By the union bound, $\mathbb{P}[(\text{rank}(C) <|I|/2) \cup (\text{rank}(C') <|I|/2)] \leq 2 \mathbb{P}[\text{rank}(C) <|I|/2]$.  Since we are assuming that $|I|$ is $\Omega(n)$ for now, \Cref{lem:corank_small} implies this event is bounded as $O(2^{-k n^2})$ for some constant $k$.  Let $U$ be the event $\left(\text{rank}(C) \geq |I|/2 \right) \cap \left(\text{rank}(C') \geq |I|/2\right)$.  Then, for large enough $n$,

\begin{align}
    \mathbb{P}\left[\textrm{Good} \right] \leq \mathbb{P}\left[\textrm{Good} |U\right]+O(2^{-kn^2}).
\end{align}

Now we will condition on the values of $\mathbf{B}$, $\mathbf{C}$, $\mathbf{C}'$ and $\mathbf{E}$.  
\begin{align}
    \mathbb{P}\left[\textrm{Good} |U\right] = \sum_{C, C', B, E} \mathbb{P}\left[\textrm{Good} \left|\begin{matrix} U, \\ \mathbf{C}=C, \\ \mathbf{C'}=C', \\ \mathbf{B}=B,\\ \mathbf{E}=E\end{matrix}\right.\right] \mathbb{P}\left[\begin{matrix}\mathbf{C}=C, \\ \mathbf{C'}=C', \\ \mathbf{B}=B,\\ \mathbf{E}=E\end{matrix}\text{\textscale{100}{|}} U\right].
\end{align}
Note that $\textrm{Lap}_S$, $\textrm{Lap}_R$ and $\textrm{Cut}_S \cap \textrm{Cut}_R$ are conditionally independent given $\mathbf{B}$, $\mathbf{C}$, $\mathbf{C}'$ and $\mathbf{E}$:  This holds because the ranks of $\mathbf{L}_S$ and $\mathbf{L}_R$ are functions of $\mathbf{A}$ and $\mathbf{F}$ respectively and $\mathbf{A}$ and $\mathbf{F}$ are independent.  Similarly, the values achieved by $\mathbf{C}_S \mathds{1}_S$ and $\mathbf{C}_R \mathds{1}_R$ (conditioned on $\mathbf{B}$ and $\mathbf{E}$) are functions of $\mathbf{D}$, $\mathbf{H}_1$, $\mathbf{H}_2$ and $\mathbf{H}_3$ which are all independent.  Hence, 

\begin{align}
    \mathbb{P}\left[\textrm{Good}\left|\begin{matrix} U, \\ \mathbf{C}=C, \\ \mathbf{C'}=C', \\ \mathbf{B}=B,\\ \mathbf{E}=E\end{matrix}\right.\right]= \mathbb{P}\left[\textrm{Lap}_S\left|\begin{matrix} U, \\ \mathbf{C}=C, \\ \mathbf{C'}=C', \\ \mathbf{B}=B,\\ \mathbf{E}=E\end{matrix}\right.\right]\cdot \mathbb{P}\left[\textrm{Lap}_R\left|\begin{matrix} U, \\ \mathbf{C}=C, \\ \mathbf{C'}=C', \\ \mathbf{B}=B,\\ \mathbf{E}=E\end{matrix}\right.\right]
    \cdot{}\mathbb{P}\left[\textrm{Cut}_S \cap \textrm{Cut}_R\left|\begin{matrix} U, \\ \mathbf{C}=C, \\ \mathbf{C'}=C', \\ \mathbf{B}=B,\\ \mathbf{E}=E\end{matrix}\right.\right].
\end{align}
By \Cref{lem:sub_laplace}, $\textrm{Lap}_S$ is equivalent to the matrix $\begin{bmatrix}\mathbf{A} & \mathbf{B} \\ \mathbf{B}^T & \mathbf{C} \end{bmatrix}$ being full rank.  Hence, we can apply \Cref{lem:worst_case_rank} to conclude that
$$
\mathbb{P}\left[\textrm{Lap}_R\left|\begin{matrix} U, \\ \mathbf{C}=C, \\ \mathbf{C'}=C', \\ \mathbf{B}=B,\\ \mathbf{E}=E\end{matrix}\right.\right],\,\,\,\,\mathbb{P}\left[\textrm{Lap}_S\left|\begin{matrix} U, \\ \mathbf{C}=C, \\ \mathbf{C'}=C', \\ \mathbf{B}=B,\\ \mathbf{E}=E\end{matrix}\right.\right] \leq q_{s-\lceil 3 |I|/2\rceil-1},
$$  
where we used the fact that $q_i \geq q_{i+1}$.  The bound so far is then:
\begin{align}
    \mathbb{P}\left[\textrm{Good} |U\right] \leq q_{s-\lceil 3 |I|/2\rceil -1}^2\sum_{C, C', B, E} \mathbb{P}\left[\textrm{Cut}_S \cap \textrm{Cut}_R\left|\begin{matrix} U, \\ \mathbf{C}=C, \\ \mathbf{C'}=C', \\ \mathbf{B}=B,\\ \mathbf{E}=E\end{matrix}\right.\right] 
    \mathbb{P}\left[\begin{matrix}\mathbf{C}=C, \\ \mathbf{C'}=C', \\ \mathbf{B}=B,\\ \mathbf{E}=E\end{matrix}\text{\textscale{100}{|}} U\right]\\
    =q_{s-\lceil 3 |I|/2\rceil-1}^2\mathbb{P}[\textrm{Cut}_S \cap \textrm{Cut}_R|U].
\end{align}

For any event $W$, if $\mathbb{P}[\neg U]$ is small then $\mathbb{P}[W|U] \leq \mathbb{P}[W] +2 \mathbb{P} [\neg U]$.  So we can bound
\begin{align}
    \mathbb{P}[\textrm{Cut}_S \cap \textrm{Cut}_R|U] = \mathbb{P}[\textrm{Cut}_S \cap \textrm{Cut}_R] +O(2^{-k n^2}).
\end{align}
Hence we need to evaluate $\mathbb{P}[\textrm{Cut}_S \cap \textrm{Cut}_R]$.  Let $\mathbf{B}'=\mathbf{Adj}[S\setminus I, I]=\mathbf{L}_S[S\setminus I, I]$ and let $\mathbf{E}'=\mathbf{Adj}[I, R\setminus I]$ (these are just $\mathbf{B}$ and $\mathbf{E}$ with the ``missing'' elements included).  $\mathbf{B}'$, $\mathbf{E}'$ and $\mathbf{D}$ are all uniform random rectangular matrices, and are all independent since we have removed any kind of conditioning.  We can write:
\begin{align}
    \mathbb{P}[\textrm{Cut}_S \cap \textrm{Cut}_R] =\sum_{v, w, x} \mathbb{P}\left[\textrm{Cut}_S \cap \textrm{Cut}_R \left| \begin{matrix} \mathbf{B}'\mathds{1}=v, \\ (\mathbf{E}')^T\mathds{1}=w, \\ \mathbf{H}_2 \mathds{1} =x\end{matrix}\right.\right] \mathbb{P}\left[\begin{matrix}\mathbf{B}'\mathds{1}=v, \\(\mathbf{E}')^T\mathds{1}=w, \\ \mathbf{H}_2 \mathds{1} =x\end{matrix}\right].
\end{align}
Conditioned on the values above, we satisfy the cut conditions exactly when $\mathbf{D}\mathds{1}=\mathds{1}+v$, $\mathbf{D}^T \mathds{1}=\mathds{1}+w, \mathbf{H}_1 \mathds{1}=\mathds{1}+x$ and $\mathbf{H}_3 \mathds{1}=\mathds{1}+x$.  $\mathbf{H}_1$, $\mathbf{H}_3$ and $\mathbf{D}$ are mutually independent so \Cref{lem:d_d_transpose} implies: 
\begin{align}
    \mathbb{P}\left[\textrm{Cut}_S \cap \textrm{Cut}_R\left| \begin{matrix} \mathbf{B}'\mathds{1}=v, \\ (\mathbf{E}')^T\mathds{1}=w, \\ \mathbf{H}_2 \mathds{1} =x\end{matrix}\right.\right] = \begin{cases}
        2^{-(2n-2s-1)} \text{ if |w|=|v|} \mod 2\\
        0 \text{  o.w.  }
    \end{cases}.
\end{align}
The values of $|v|$ and $|w|$ are uniform and independent so we can compute $\mathbb{P}[\textrm{Cut}_S \cap \textrm{Cut}_R] = 2^{-(2n-2s)} $.  
Now we can bring everything together.  
\begin{align*}
    \mathbb{P}[\textrm{Good}] \leq \mathbb{P}\left[\textrm{Good} |U\right]+O(2^{-kn^2})= O(2^{-k n^2})+\sum_{C, C', B, E} \mathbb{P}\left[\textrm{Good} \left|\begin{matrix} U, \\ \mathbf{C}=C, \\ \mathbf{C'}=C', \\ \mathbf{B}=B,\\ \mathbf{E}=E\end{matrix}\right.\right] \mathbb{P}\left[\begin{matrix}\mathbf{C}=C, \\ \mathbf{C'}=C', \\ \mathbf{B}=B,\\ \mathbf{E}=E\end{matrix}\text{\textscale{100}{|}} U\right]\\
    \leq O(2^{-k n^2})+ q_{s-\lceil 3 |I|/2\rceil-1}^2 \sum_{C, C', B, E} \mathbb{P}\left[\textrm{Cut}_S \cap \textrm{Cut}_R\left|\begin{matrix} U, \\ \mathbf{C}=C, \\ \mathbf{C'}=C', \\ \mathbf{B}=B,\\ \mathbf{E}=E\end{matrix}\right.\right]\mathbb{P}\left[\begin{matrix}\mathbf{C}=C, \\ \mathbf{C'}=C', \\ \mathbf{B}=B,\\ \mathbf{E}=E\end{matrix}\text{\textscale{100}{|}} U\right]\\
    = O(2^{-k n^2})+q_{s-\lceil 3 |I|/2\rceil-1 }^2\mathbb{P}\left[\textrm{Cut}_S \cap \textrm{Cut}_R|U\right]\\
    \leq O(2^{-k n^2})+q_{s-\lceil 3 |I|/2\rceil -1}^2\mathbb{P}\left[\textrm{Cut}_S \cap \textrm{Cut}_R\right]\\
    =O(2^{-k n^2})+q_{s-\lceil 3 |I|/2\rceil-1}^2 2^{-(2n-2s)}.
\end{align*}

Now we will establish bounds for the ``non-typical'' cases.  If $S=R$ or if $S\cap R =\emptyset$ then $\mathbb{P}[\textrm{Good}] \leq \mathbb{P}[\textrm{Cut}_S]=2^{-(n-s)} $ since $\mathbf{C}_S$ is a uniform random $(n-s) \times s$ matrix in this case.  If $0 < |I|< n/2$, but $I$ is outside of the ``typical'' interval then the analysis for $\mathbb{P}[\textrm{Cut}_S \cap \textrm{Cut}_R]$ above still holds so we can derive $\mathbb{P}[\textrm{Good}] \leq \mathbb{P}[\textrm{Cut}_S \cap \textrm{Cut}_R] = 2^{-(2n-2s)}$.  The general bound we derive is:
\begin{align}
    \mathbb{P}[\textrm{Good}]\leq \begin{cases}
    O(2^{-k n^2})+q_{s-\lceil 3 |I|/2\rceil-1 }^2 2^{-(2n-2s)} \text{ if $|I|\in [n/4-\delta n, n/4+\delta n]$}\\
        2^{-(n-s)}\leq 2^{-n/2} \text{ if $|I|=0$ or $|I|=s$}\\
        2^{-(2n-2s)}\leq 2^{-n} \text{ o.w. }
    \end{cases}
\end{align}

Now we will turn to an upper bound for $\mathbb{E}[\mathbf{N}^2]$.  Let us define $f(\Delta)$ as the probability of $\textrm{Good}$ when $S$ and $R$ have overlap size $\Delta$.  
\begin{align}\label{eq:bound_on_N2}
    \nonumber \mathbb{E}[\mathbf{N}^2]=\binom{n}{s}\sum_{\Delta=0}^{s} f(\Delta)\binom{s}{\Delta} \binom{n-s}{s-\Delta} \\
    \nonumber \leq \binom{n}{s} \bigg(O(n2^{-n/2}) + \sum_{\Delta \in[n/4-\delta n, n/4+\delta n]\cap \mathbb{Z}} (O(2^{-k n^2})+q_{s-\lceil 3 |I|/2\rceil-1}^2 2^{-(2n-2s)})\binom{s}{\Delta} \binom{n-s}{s-\Delta}   \\
    +2^{-n}\sum_{\Delta \notin \{0, n/2,[n/4-\delta n, n/4+\delta n]\cap \mathbb{Z}\}}\binom{s}{\Delta} \binom{n-s}{s-\Delta} \bigg).
\end{align}
Well-known bounds \cite{CHVATAL1979285} establish
\begin{align}
    \sum_{\Delta\notin [n/4-\delta n, n/4 + \delta n]\cap \mathbb{Z}}  \binom{s}{\Delta} \binom{n-s}{s-\Delta} = O(2^{n(1-4\delta^2)}).
\end{align}
Note that $q_i \geq q_{i+1}$ so in the typical region we can uniformly upper bound $q_{s-\lceil 3 |I|/2\rceil-1  } \leq q_{\lfloor n(1/8-3 \delta /2)-3 \rfloor }=q_{\lfloor n/16 -3\rfloor}\leq q_{\lfloor n/17\rfloor}$ for large enough $n$.  Using Vandermonde's identity and plugging these facts into \Cref{eq:bound_on_N2} we obtain:

\begin{align}
    \nonumber \mathbb{E}[\mathbf{N}^2] \leq \binom{n}{s}\left( O(n2^{-n/2})+q_{\lfloor n/17\rfloor}^2 2^{-(2n-2s)} \binom{n}{s}+O(2^{-4 \delta^2 n})\right)\\
    =\binom{n}{s} \left( O(2^{-4 \delta^2 n})+q_{\lfloor n/17\rfloor}^2 2^{-(2n-2s)} \binom{n}{s}\right).
\end{align}

Now we can compute
\begin{align}
    \frac{\mathbb{E}[\mathbf{N}]^2}{\mathbb{E}[\mathbf{N}^2]} \geq \frac{\left(\binom{n}{s}2^{-\lceil n/2\rceil }q_{s-1}\right)^2}{\binom{n}{s} \left( O(2^{-4 \delta^2 n})+q_{\lfloor n/17\rfloor}^2 2^{-(2n-2s)} \binom{n}{s}\right)} \geq \frac{q_{s-1}^2}{q_{\lfloor n/17\rfloor }^2}-O(\sqrt{n}2^{-4 \delta^2 n})
\end{align}

By \Cref{lem:q_convergence}, both $s-1$ and $\lfloor n(1/8-3\delta/2)-3\rfloor$ are linear in $n$, so $q_{s-1}=q_\infty\bc{1+\Theta(2^{-n/2})}$ and $q_{\lfloor n/17\rfloor}=q_\infty\bc{1+\Theta(2^{-n/17})}$.  Thus 
\[
    \frac{q_{s-1}^2}{q_{\lfloor n/17\rfloor}^{2}}=\frac{q_\infty^2\bc{1+\Theta(2^{-n/2})}}{q_\infty^2\bc{1+\Theta(2^{-n/17})}}=1-O(2^{-n/17}).
\]
Paley-Zygmund gives
\[
    \PP\brk{\mathbf N\geq 1} \geq \frac{\EE\brk{\mathbf N}^2}{\EE\brk{\mathbf N^2}}=1-O(\sqrt{n}2^{-n/144}).
\]
\end{proof}

\end{document}